\PassOptionsToPackage{table}{xcolor}

\documentclass[letterpaper,twocolumn,10pt]{article}
\usepackage{usenix-2020-09}

\usepackage[utf8]{inputenc}
\usepackage{textcomp}
\usepackage{amsmath, amsfonts, amsthm}
\usepackage{pifont}

\usepackage{xspace}
\usepackage{framed}

\usepackage{tikz}
\usetikzlibrary{arrows, positioning, calc, decorations,shapes,decorations.pathreplacing,chains,fit,arrows.meta}

\usepackage[margin=.05\linewidth]{subfig}
\usepackage{caption}
\usepackage{color}
\usepackage{import}

\newcommand{\pyes}{\ding{108}}%
\newcommand{\psome}{\ding{119}}%
\newcommand{\pno}{\ding{109}}%
\definecolor{lightlightgray}{rgb}{0.95, 0.95, 0.95}
\newtheorem{definition}{Definition}
\newtheorem{lemma}{Lemma}
\tikzset{%
  fontscale/.style={font=\relsize{#1}},
  >=stealth',
  box/.style={
    rectangle,
    draw=black,
    text width=8em,
    minimum height=2em,
  },
  small box/.style={
    rectangle,
    draw=black,
    text width=2.5em,
    minimum height=2.5em,
    text centered
  },
  small box nc/.style={
    rectangle,
    draw=black,
    text width=2.5em,
    minimum height=2.5em
  },
  arrow/.style={
    ->,
    shorten <=2pt,
    shorten >=2pt,
  }
}

\newcommand\lbox[2]{#1\\\vspace{-0.5em}\hrulefill\\\footnotesize{#2}}
\newcommand\lboxNR[1]{#1}
\newcommand\arrowt[1]{\footnotesize{#1}}

\tikzset{>=latex}

\pgfdeclareimage[width=1.7cm]{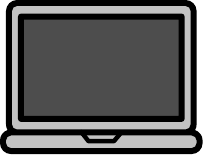}{figures/icons/client}
\tikzstyle{client-iconspace} = [rectangle,minimum width=1.5cm]
\pgfdeclareimage[width=.8cm]{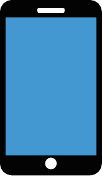}{figures/icons/device}
\tikzstyle{device-iconspace} = [rectangle,minimum width=.8cm]
\pgfdeclareimage[width=1.3cm]{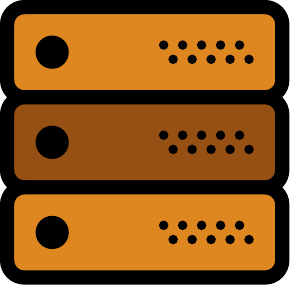}{figures/icons/cloud}
\tikzstyle{cloud-iconspace} = [rectangle,minimum width=1cm]
\pgfdeclareimage[width=1.5cm]{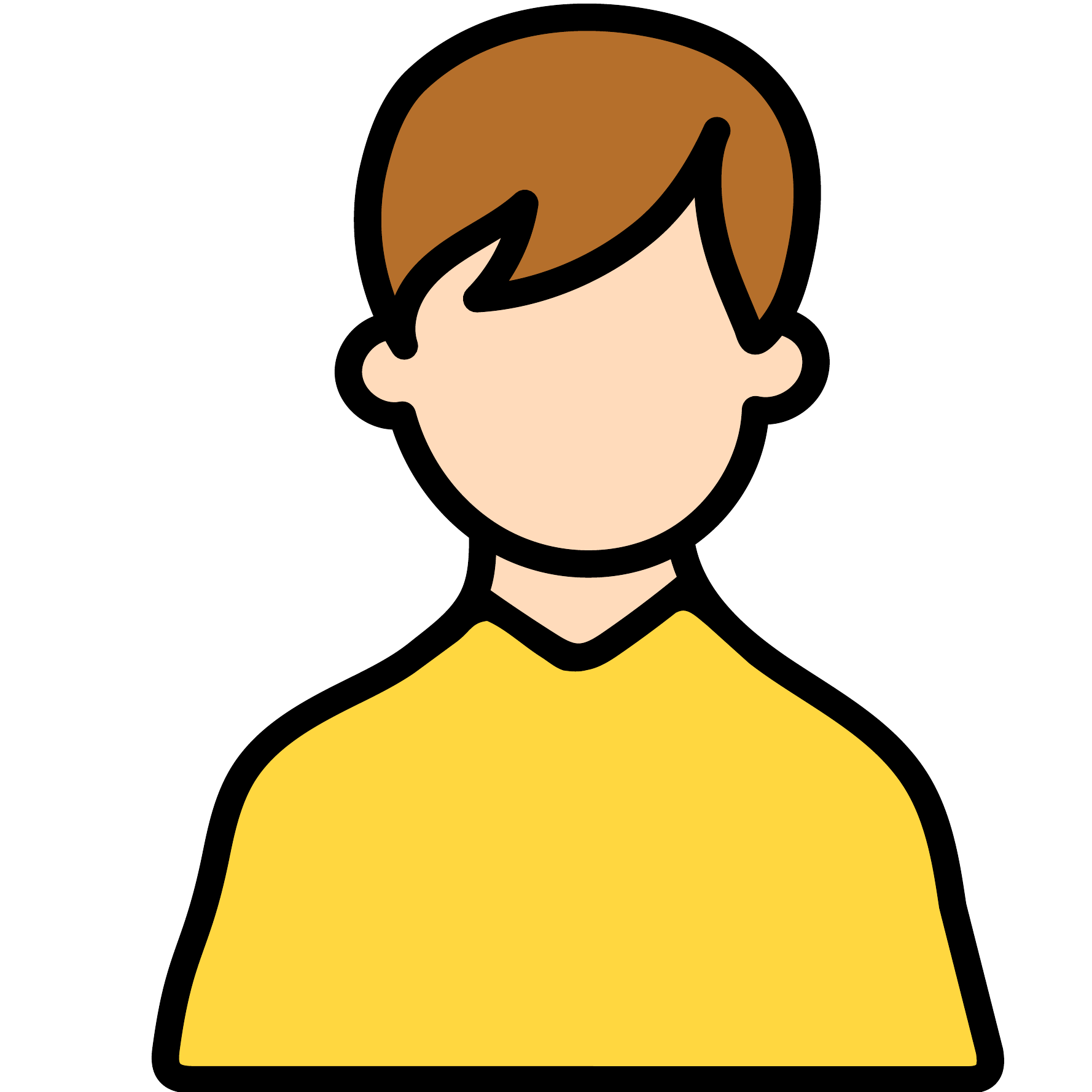}{figures/icons/user}
\tikzstyle{user-iconspace} = [rectangle,minimum width=1.8cm]

\tikzstyle{icon}=[rectangle,draw=white!50,fill=white,align=center,minimum height=1em,minimum width=2em]
\tikzstyle{iconspace}=[rectangle,draw=white!50,fill=white,align=center,minimum width=2em]
\tikzstyle{rect}=[rectangle,draw=black!50,thick,fill=white,align=center,minimum height=0.75cm,minimum width=2cm]
\tikzstyle{rectnone}=[rectangle,draw=none,thick,fill=white,align=center,minimum height=2.2cm,minimum width=1.8cm]
\tikzstyle{txt}=[rectangle,draw=white,fill=white,align=center,minimum width=2cm]

\newcommand{\name}{2FE\xspace}
\newcommand{\walletname}{2FW\xspace}

\mathchardef\mhyphen="2D

\newcommand\adv{\ensuremath{\mathcal{A}}}

\newcommand\set[1]{\ensuremath{\lbrace #1 \rbrace}}

\newcommand\Z{\ensuremath{\mathbb{Z}}} 
\newcommand\bits{\ensuremath{\{0,1\}}} 
\newcommand\pick{\ensuremath{\xleftarrow{\text{\tiny{\$}}}}} 
\newcommand\Enc{\ensuremath{\mathsf{Enc}}} 
\newcommand\Dec{\ensuremath{\mathsf{Dec}}} 

\newcommand\SR{\ensuremath{\mathsf{SR}}}
\newcommand\tPRF{\ensuremath{\mathsf{tPRF}}}
\newcommand\tSIG{\ensuremath{\mathsf{tSIG}}}

\newcommand\Sig{\ensuremath{\mathsf{Sig}}}
\newcommand\Ver{\ensuremath{\mathsf{Ver}}}

\newcommand\cmd[1]{\texttt{(}\ensuremath{\mathtt{#1}}\texttt{)}}
\newcommand\cmda[2]{\texttt{(}\ensuremath{\mathtt{#1},#2}\texttt{)}}

\newcommand\ctxt{\ensuremath{\mathsf{ctxt}}}

\newcommand\C{\ensuremath{\mathsf{P}}}  
\newcommand\D{\ensuremath{\mathsf{S}}}  
\newcommand\U{\ensuremath{\mathsf{U}}}  
\renewcommand\S{\ensuremath{\mathsf{C}}}  

\newcommand\PRFkey{\ensuremath{\mathsf{K}}}
\newcommand\keyC{\ensuremath{\PRFkey_{\C}}}
\newcommand\keyD{\ensuremath{\PRFkey_{\D}}}

\newcommand\Cshare[1]{\ensuremath{u_{#1}}}
\newcommand\Dshare[1]{\ensuremath{v_{#1}}}

\newcommand\nizkDDH{\ensuremath{\mathsf{NIZK}}}
\newcommand\TPPRF{\ensuremath{\mathsf{2PPRF}}}
\newcommand\RF{\ensuremath{\mathsf{RF}}}  
\newcommand\TG{\ensuremath{\mathsf{TG}}}  

\newcommand\pk{\ensuremath{\mathsf{pk}}}
\newcommand\sk{\ensuremath{\mathsf{sk}}}

\newcounter{protc}
\renewcommand{\theprotc}{\arabic{protc}}
\newenvironment{protocol}[2]{%
  \refstepcounter{protc}%
  \label{#1}%
  \begin{framed}
    \noindent\textbf{Protocol~\theprotc.} #2\\
    \rule{\textwidth}{0.2pt}
  }%
  {
    \end{framed}%
  }

\newcounter{funcc}
\renewcommand{\thefuncc}{\arabic{funcc}}
\newenvironment{functionality}[2]{%
  \refstepcounter{funcc}%
  \label{#1}%
  \begin{framed}
    \noindent\textbf{Functionality~\thefuncc.} #2\\
    \rule{\textwidth}{0.2pt}
  }%
  {
    \end{framed}%
  }

\usepackage[compact]{titlesec}
\titlespacing{\section}{0pt}{3ex}{1ex}
\titlespacing{\subsection}{0pt}{2ex}{0.5ex}
\titlespacing{\subsubsection}{0pt}{0.5ex}{0.5ex}

\let\markeverypar\everypar
\newtoks\everypar
\everypar\markeverypar
\markeverypar{\the\everypar\looseness=-1\relax}
\renewcommand{\paragraph}[1]{\vspace{2pt}\textbf{#1}}

\let\labelindent\relax
\usepackage{booktabs}
\usepackage{tabularx}

\usepackage{ragged2e}
\usepackage{multirow}
\usepackage{makecell}
\newcommand{\myrot}[1]{\rotatebox{90}{\makecell{#1}}}
\usepackage{paralist}
\usepackage[shortlabels]{enumitem}

\begin{document}
	
\date{}

\title{2FE: Two-Factor Encryption for Cloud Storage}

\author{
	{\rm Anders P. K. Dalskov}\\
	Aarhus University\\
	\and
	{\rm Daniele Lain}\\
	ETH Zurich
	\and
	{\rm Enis Ulqinaku}\\
	ETH Zurich
	\and
	{\rm Kari Kostiainen}\\
	ETH Zurich
	\and
	{\rm Srdjan Capkun}\\
	ETH Zurich
} 

\maketitle

\begin{abstract}
	Encrypted cloud storage services are steadily increasing in popularity, with many commercial solutions currently available. In such solutions, the cloud storage is trusted for data availability, but not for confidentiality. Additionally, the user’s device is considered secure, and the user is expected to behave correctly. 
We argue that such assumptions are not met in reality: e.g., users routinely forget passwords and fail to make backups, and users’ devices get stolen or become infected with malware. Therefore, we consider a more extensive threat model, where users’ devices are susceptible to attacks and common human errors are possible. Given this model, we analyze 10 popular commercial services and show that none of them provides good confidentiality and data availability.

Motivated by the lack of adequate solutions in the market, we design a novel scheme called \emph{Two-Factor Encryption} (2FE) that draws inspiration from two-factor authentication and turns file encryption and decryption into an interactive process where two user devices, like a laptop and a smartphone, must interact. 2FE provides strong confidentiality and availability guarantees, as it withstands compromised cloud storage, one stolen or compromised user device at a time, and various human errors. 
2FE achieves this by leveraging secret sharing with additional techniques such as oblivious pseudorandom functions and zero-knowledge proofs. We evaluate 2FE experimentally and show that its performance overhead is small.
Finally, we explain how our approach can be adapted to other related use cases such as cryptocurrency wallets.


\end{abstract}

\captionsetup[table]{skip=0pt}
\captionsetup[figure]{skip=5pt}

\section{Introduction}
\label{sec:introduction}

Personal cloud storage provides users an easy way to access and backup their data. While the major service providers are usually benign and reputable companies, online services often get compromised by external attackers and malicious insiders, which means that users' files can be leaked. Another risk that such service face is that law enforcement might coerce the service providers into exposing users' files.

Encrypted cloud storage systems attempt to address such data confidentiality concerns by allowing users to encrypt their files before uploading them to the cloud. The market of encrypted cloud storage is steadily growing, e.g., a service called Mega~\cite{mega-whitepaper} recently reported 130 million users~\cite{megausers}.

Most currently available solutions consider a model where the cloud storage is trusted for data availability, but not for data confidentiality. Such solutions (often implicitly) assume that the users' devices can be trusted for storing user credentials and encryption keys. Furthermore, such solutions assume that users pick passwords with sufficient entropy, do not forget their passwords, or take backups of their encryption keys.

We argue that such assumptions are not met in practice. Industry reports state that thousand of (corporate) laptops are stolen or go missing daily~\cite{ponemon-billion, ponemon-airport}---and many data breaches are a result of device theft~\cite{bitglass}. Infection by malware is common place in today's PC and smartphone platforms. Users are also known to pick passwords of low entropy, reuse them across services and forget their passwords~\cite{bonneau2012science,florencio2007large,wash2016understanding}. 

\paragraph{Analysis of existing services.} To enable more precise analysis of such solutions, we define an extensive threat model for encrypted cloud storage that extends the (often implicitly assumed) setting of commercial solutions in two ways. First, we assume that users' devices may get stolen or infected with malware, and that adversaries may get temporary access to them. Second, we assume that users enable ``convenience features'' like long-lasting sessions, pick insecure passwords, forget passwords, and fail to back up encryption keys.

Given such threat model, we analyze 10 popular encrypted cloud storage providers on the market (see Table~\ref{tab:comparison}). The available commercial solutions follow two main approaches. The first is to derive the encryption key from a user-chosen password. Unsurprisingly, we find that all such solutions fail to provide sufficient entropy to prevent simple brute-force attacks and all user's data is irrecoverably lost if the user forgets the password.
The second main approach is to encrypt files using a full-length pseudorandom key that is stored on the user's device. As expected, such solutions fail to provide confidentiality when the user's device is compromised. Lost or stolen user device also means loss of all data availability.

Our analysis uncovers also more subtle and unexpected problems. For example, we find that some services combine password-based encryption with long-lived sessions. Probably the rationale of introducing such features was to improve usability, but it turns out that such solutions combine the main weaknesses of both approaches: encrypted files can be brute-forced by the untrusted cloud and also read by a party that steals the user's device. 
We conclude that there is no encrypted cloud storage service on the market that would provide good data confidentiality and availability under a realistic threat model.

\paragraph{Our solution.} Motivated by the poor state of affairs in encrypted cloud storage, we propose a new scheme called \emph{Two-Factor Encryption} (\name)
for scenarios where the user has two devices, such as a laptop and smartphone. We denote the device that is used to encrypt and decrypt files as \emph{primary}, and the assisting device \emph{secondary}.
We rule out possible designs that duplicate encryption keys to both devices, because such solutions do not tolerate theft or malware infection on either device. Instead, we take \emph{secret-sharing} as the starting point of our solution, as it allows us to build a more robust solution where one of the two user devices may be exposed to the adversary at a time. 

We observe that a simple secret-sharing scheme, where a master encryption key is split into two shares that are stored on primary and secondary, is insufficient to solve our problem. One challenge is that a lost or stolen device (and thus one lost key share) would mean that the user can no longer recover the full encryption secret. We address this challenge by using the untrusted cloud storage as a third ``virtual device'' which stores a recovery key share that is released to the user after successful identity verification (which may happen out-of-band) in case of device theft or permanent loss. 

Another challenge is a compromised secondary device that participates in the process of deriving per-file encryption keys. A malicious secondary may attempt to learn primary's key share or the resulting encryption key which would violate data confidentiality. Alternatively, a compromised secondary may manipulate the key derivation process such that reconstruction of the same encryption key is not possible afterwards which would violate data availability. To address such problems, we use Oblivious Pseudo-Random Functions \cite{10.1007:978-3-540-30576-7_17} in the per-file encryption key derivation process and complement them with simple zero-knowledge proofs that makes detection of misbehaving secondary possible on the primary device. 
We explain further challenges and how to address them in Section~\ref{sec:challenges}.
 
Our solution, \name, provides strong data confidentiality and availability guarantees. The encrypted files cannot be revealed using brute-force attacks, because each per-file encryption key is sufficiently long and pseudorandom. 
Theft or compromise of one of the user devices is also insufficient to read the encrypted files, because the used per-file encryption keys cannot be reconstructed using a single key share alone. \name provides also good data availability, because all encrypted data remains accessible to the user even if one of the two devices is lost or stolen, or if the user forgets his authentication credentials. To the best of our knowledge, \name is the first and only encrypted cloud storage solution that provides such strong confidentiality and availability, as shown by Table~\ref{tab:comparison}.

The basic variant of \name requires that the user owns and enrolls two devices, but only interacts the primary during file encryption and decryption which preserves the user experience of standard cloud storage. For additional protection against infected primary, our solution can be complemented with notifications or prompts on the secondary device that allow the user to detect or prevent unauthorized file decryption. 

We implemented a prototype of \name and evaluated it using non-encrypted Dropbox storage as a baseline. Our measurements show that added delay on file upload and download times is small (200-380 ms). Compared to non-encrypted file upload and download, the performance overhead is moderate (e.g., 5-20\%), depending on the size of the file.

Finally, we explain that the same approach can be adapted to other use cases as well. To show one such example, we present a novel design for a cryptocurrency wallet, called \emph{Two-Factor Wallet} (2FW), that provides improved confidentiality and availability compared to existing wallet solutions.

\paragraph{Contributions and roadmap.} To summarize, this paper makes the following contributions: 

\begin{itemize} 
	\item \emph{Analysis of commercial services.} We define an extensive but realistic threat model for encrypted cloud storage and analyze 10 popular commercial services using it. We find that none of them provides satisfactory data confidentiality and availability (Section~\ref{sec:motivation}). 
	
	\item \emph{Encrypted cloud storage scheme.} We present a novel encrypted cloud storage solution called Two-Factor Encryption (\name) that combines secret-sharing with additional protections such as oblivious pseudo-random functions and zero-know proofs (Section~\ref{sec:2fe_overview}). We show that \name achieves significantly better data confidentiality and availability compared to any previous solution (Section~\ref{sec:security-analysis}).
	
	\item \emph{Implementation and evaluation.} We evaluate \name experimentally and demonstrate that its performance overhead on file encryption and decryption is small (Section~\ref{sec:evaluation}).
	
	\item \emph{Wallet adaptation}. We show that our approach can be adapted to other use cases. As an example we describe a cryptocurrency wallet, Two-Factor Wallet (2FW), that provides strong confidentiality and high availability (Section~\ref{sec:cw}).
\end{itemize}


\section{Analysis of Commercial Solutions}
\label{sec:motivation}

Over the recent years, several \emph{personal} encrypted cloud storage services have become available. However, the security of such solutions is not apparent, because most service providers do not publish detailed security analysis that would specify the considered threats and provided security properties. 
Our first goal in this paper is to improve the understanding about the security of such solutions. To achieve this, we establish an analysis criteria for encrypted cloud storage solutions that includes precise system and threat models and desirable security properties. After that, we analyze 10 popular service based on this criteria.
We note that related \emph{corporate} cloud encryption solutions, where the users have trusted support infrastructure and personnel available, are beyond the focus of this paper.

\subsection{System Model}
\label{sec:motivation.system}

The system model for personal encrypted cloud storage consists of two main entities, user and storage provider, as illustrated in Figure~\ref{fig:system_problem_statement}.

\paragraph{User.} The user owns one or more devices, such as laptop and smartphone, and uses one of them to encrypt and upload files for cloud storage, and later download and decrypt files. Occasionally, the user needs to replace one of his devices with a new one. The user picks a password of low entropy~\cite{florencio2007large,wash2016understanding} and he may reuse passwords across different cloud services~\cite{bonneau2012science, das2014passreuse}. The user may also forget the passwords he chooses. The user does not necessarily perform offline backup of files that are stored on his devices. If the service provides convenience options, such ``remember me'' feature that saves a session token on the device, the user will enable them.

\paragraph{Storage provider.}
The storage provider runs three services. The first is standard \emph{user authentication}, e.g., based on username and password. The user's  authentication credential can be reset, e.g., using email.
The second is \emph{file storage} that allows users to upload arbitrary (encrypted) files, as well as retrieve, and delete them.
Deleted files are kept in a ``trash bin'' for some time, to allow users to change their mind and restore them. File retrieval requires user authentication.
The third service is \textit{identity verification}. In case the user gets locked out of his account (e.g., forgets his password and can no longer access his email), the user can prove his identity to the service provider through other means. Common examples of identity verification include sending a copy of their ID to the service provider or showing it over a video link. Most major online service providers have such out-of-band identity verification mechanisms in place~\cite{google-app,microsoft-recovery}. 

\begin{figure}[t]
	\centering
	\fontsize{10pt}{14pt}\selectfont
	\def\svgwidth{\linewidth}
	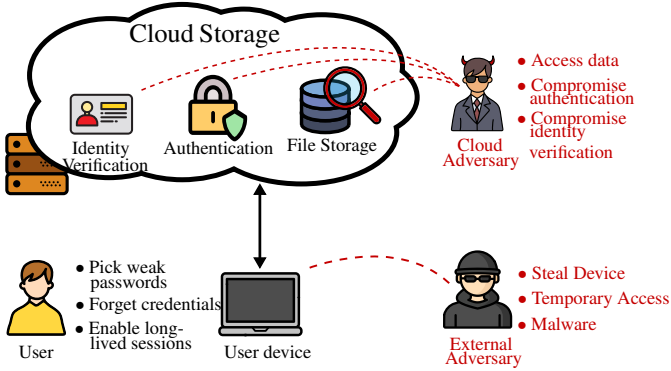
	\caption{System model and threat model for personal encrypted cloud storage.}
	\label{fig:system_problem_statement}

\end{figure}

\subsection{Threat Model}
\label{sec:motivation.threat}

The encrypted cloud storage solution faces two types of threats, as shown in  Figure~\ref{fig:system_problem_statement}. The first set of threats is related to the cloud storage and we model them as \emph{cloud adversary}. The second set of threats is related to the user's devices and we model them as \emph{external adversary}. These two adversaries are separate entities and they do not collude.

\paragraph{Cloud adversary.} The company running the cloud storage service is considered reputable and benign. However, the company might experience an external attack, it might have a malicious insider (such as bad administrator), or it might be coerced to release any data that it stores (e.g., request by law enforcement, subpoena). To model such threats, we say that the cloud adversary has unrestrained access to the file storage (data at rest) and he can tamper with the execution of any protocol with the user, including the user authentication and identity verification. The client software distributed by the provider to end-users is benign (and can be examined by third parties).

\paragraph{External adversary.} While most commercial solutions assume an untrusted file storage on the cloud, they typically implicitly assume that the user's devices are trusted. However, in reality users often lose their devices~\cite{ponemon-airport}, and thousands of laptops are stolen or go missing daily~\cite{ponemon-billion} or get accessed by third parties when left unlocked~\cite{ponemon-billion}. Malware is also prevalent in commodity computing devices such as PCs and smartphones. We consider such threats and define that the external adversary may have the following capabilities:

\begin{itemize}
    \item \emph{Stolen device.}
The external adversary can steal the user's device, and thus access all the files on its hard drive. We assume access to any data at rest, because techniques such as full-disk encryption are not yet in prevalent use, especially outside the corporate domain.%
\footnote{Full-disk encryption is an option feature in popular consumer PC platforms like Windows 10. Further, some Windows instances enable password-less full-disk encryption only, and thus the attacker would only need to turn on the device to get the disk decrypted. Because of the high risk associated with the loss of the disk password, it is unlikely that password-protected full-disk encryption will by enabled by default for all users.}

    \item \emph{Temporary access.}
The external adversary may also gain temporary access to one of the user's devices that was left unlocked or circumvent the used screen lock protection~\cite{ponemon-billion} --- also known as ``evil maid attacks''~\cite{bellare1999non}. With temporary access, the external adversary can use the device on behalf of the user for a short amount of time (more advanced scenarios, such as cold-boot attacks~\cite{halderman2009lest} to extract all RAM content, are out of scope).

\item \emph{Malware infection.}
The external adversary may also be able to infect one of user's devices with malware. This is the most powerful scenario, where the external adversary has access to all files stored on the hard drive and all data in memory (such as open session tokens), and he learns any user input that is provided to the device (e.g., user's password).
\end{itemize}

\subsection{Security Properties}
\label{sec:analysis:properties}

The main security properties of encrypted cloud storage are data confidentiality and availability.

\paragraph{Confidentiality} means that only the owner of the file (i.e., the user) can read its content. With respect to the cloud adversary, confidentiality should hold at all times. When considering the external adversary, we are interested in the confidentiality of files that were uploaded to the cloud storage \textit{before} the device was compromised---clearly, files on the stolen or compromised device cannot have any confidentiality guarantee.

\paragraph{Availability} means that the user can maintain access to files stored on the cloud at all times. Availability with respect to the cloud provider is out of scope: a malicious, compromised, or coerced provider can deny service to the users or delete all users' files. In practice, this is not a major threat, though, because a commercial service provider would have consequences if it loses all users' data.
Similar to above, we are interested in the availability of files that were stored on the cloud \emph{before} the adversary compromised the user's device. Obviously files uploaded during malware infection cannot guarantee availability, as the external adversary could, for example, replace them with arbitrary data.

\paragraph{Non-requirements.} We do \emph{not} consider more advanced confidentiality requirements, such as hiding users file access patterns or confidentiality of encrypted file sizes. Some of such requirements can be achieved through orthogonal security techniques such as private information retrieval~\cite{goldreich1996software} and meta-data hiding encryption~\cite{nikitin2019purbs}.
Other advanced cloud storage properties, such as  proof-of-possession~\cite{ateniese2007provable} and deduplication for encrypted data~\cite{storer2008secure, Liu:CCS15, Bellare:Sec13} are also out of scope.

\subsection{Analysis Method}
\label{sec:analysis:testing}

We examined 10 of the most popular encrypted cloud storage services available on the market which are listed in Table~\ref{tab:comparison}. 
As the starting point our analysis, we took any available white-paper.
If the application supports a web interface, then we use the in-browser debugger to examine what data is transmitted to the service.
If the application in addition supports a desktop application, we install it and examine the local configuration files that the application uses and its operational pattern.
For example, if the application starts automatically after a system reboot, it means the user's authentication credentials and keys must be stored in the clear somewhere on the user's device.
When a particular feature is not mentioned by the provider (e.g., not in the white-paper or online documentation, and we do not find any evidence of its implementation from our manual inspection), we assume the  service in question does not support it.


{
\rowcolors{1}{white}{lightlightgray}
\def\arraystretch{1.1} 
\begin{table}[t]
	\footnotesize
	\caption{Comparison of encrypted cloud storage solutions. We mark \pyes~for fully provided, \psome~for partially-provided, and \pno~for not supported security property.}
	\label{tab:comparison}
	\begin{tabularx}{\linewidth}{lXXXXXXXX}
		\toprule\hiderowcolors 
		& \multicolumn{4}{c}{Confidentiality} & \multicolumn{4}{c}{Availability} \\
		\cmidrule(rl){2-5} \cmidrule(rl){6-9}
		& \myrot{Cloud\\Provider} & \myrot{Stolen\\Device} & \myrot{Temporary\\Access} & \myrot{Malware} & \myrot{Stolen\\Device} & \myrot{Temporary\\Access} & \myrot{Malware} & \myrot{Forgotten\\Password} \\
		
		\midrule\showrowcolors
		
		\multicolumn{9}{l}{\scriptsize{\textit{Encryption using pseudorandom key}}} \\
		Tarsnap~\cite{tarsnap} & \pyes & \pno & \pno & \pno & \pno & \pno & \pno & \pyes \\
		
		\midrule \midrule
		
		\multicolumn{9}{l}{\scriptsize{\textit{Encryption using password}}} \\
		
		Mega~\cite{mega-security} & \pno & \pno & \pno & \pno & \pno & \pno & \pno & \pno \\
		
		pCloud~\cite{pcloud-encrypted-cloud-storage} & \pno & \pyes & \pno & \pno & \pyes & \pno & \pno & \pno \\
		
		Sync~\cite{sync-your-privacy} &\pno & \pno & \pno & \pno & \psome & \psome & \psome & \psome \\
		
		Woelkli~\cite{woekli} & \pno  & \pno & \pno & \pno & \pno & \pno & \pno & \psome \\
		
		SpiderOak~\cite{spideroak-no-knowledge} &  \pno & \pno & \pno & \pno & \pno & \pno & \pno & \pno \\
		
		Tresorit~\cite{tresorit-security} & \pno & \pno & \pno & \pno & \pno & \pno & \pno & \pno \\
		
		Boxcryptor~\cite{boxcryptor} & \pno & \pno & \pno & \pno & \psome & \psome & \psome & \psome \\
		
		Zoolz~\cite{zoolz} & \pno & \pno & \pno & \pno & \pno & \pno & \pno & \pno \\
				
		IDrive~\cite{idrive} & \pno & \pno & \pno & \pno & \pno & \pno & \pno & \pno \\
		
		\midrule \midrule
		
		\multicolumn{9}{l}{\scriptsize{\textit{Two-factor solutions}}} \\
		
		Key duplication & \pyes & \pno & \pno & \pno & \pno & \pno & \pno & \pyes \\
		
		Simple secret sharing & \pyes & \pyes & \pno & \pno & \pno & \pno & \pno & \pyes \\		
		
		\textbf{2FE} & \pyes & \pyes & \psome & \psome & \pyes & \pyes & \psome & \pyes \\
		
		\textbf{2FE with prompts} & \pyes & \pyes & \pyes & \pyes & \pyes & \pyes & \psome & \pyes \\
		
		\bottomrule
		
	\end{tabularx}\vspace{-1em}
\end{table}
}

\subsection{Analysis Findings}
\label{sec:analysis:findings}

The analyzed commercial services fall into two main categories: (i)~solutions that encrypt files using a \emph{full-length pseudorandom key}; and (ii) systems that derive the file encryption key from a \emph{user-chosen password}. Below we explain the main security guarantees and weaknesses of both of these categories.  Appendix~\ref{sec:analysis-remarks} provides further remarks on the analyzed services.

\paragraph{Encryption using pseudorandom key.} Out of the 10 analyzed solutions, only one, Tarsnap~\cite{tarsnap}, leverages a full-length pseudorandom key for encryption. 
As shown in  Table~\ref{tab:comparison}, this approach provides strong confidentiality protection against the cloud adversary. Unfortunately, there are no confidentiality guarantees against the external adversary. If the external adversary manages to steal the user's device, access it temporarily, or infect it with malware, he can read all user's files. In such cases, ciphertext retrieval is possible for the adversary using an open session token saved on the device. 

This approach also suffers from poor availability. All encrypted files become irrecoverable if the user's device is lost or stolen, as the user loses the only copy of the key. The same holds for temporary access and malware: while directly deleting files can be mitigated with the help of a server-side ``trash bin'', the adversary can simply delete the key from the device. Data availability is guaranteed if the user forgets his password that is used as authentication credential, as standard password reset mechanism can be used.

\paragraph{Encryption using password.} Most services (9 out of 10) protect encrypted files using a user-chosen password, as show in Table~\ref{tab:comparison}. Such solutions can be further classified in to two groups. In the first group we have solutions, such as Mega~\cite{mega-security} and pCloud~\cite{pcloud-encrypted-cloud-storage}, where the user-chosen password is used as a seed to generate per-file encryption keys. The second group consists of systems like Sync~\cite{sync-your-privacy}, SpiderOak~\cite{spideroak-no-knowledge} and Tresorit~\cite{tresorit-security}, where files stored on the cloud are encrypted with a full-length pseudorandom key that, in turn, is encrypted with the user's password. The encrypted keys are uploaded on the cloud storage as well.

The main shortcoming of all password-based systems is lack of confidentiality towards the cloud adversary that has direct access to all the ciphertexts. When a user-chosen password is used for online authentication, the server can easily limit the number of password guesses. But when the same password is used for encryption, and the adversary has access to the ciphertext, breaking the encryption through brute-force attacks becomes often possible. An average password provides no more than 20 bits of security against offline guessing~\cite{bonneau2012science, florencio2007large,wash2016understanding}.
This applies even when password composition policies are in place~\cite{komanduri2011passwords}. 

For user convenience, most of the analyzed solutions provide long-lived sessions that last across device reboots. Such a design choice requires storing either authentication credentials or active session token on the hard drive of the user's device. Thus, if the external adversary steals the user's device, he may use the authentication credential or the session token to retrieve ciphertexts and then break the encryption, as discussed above or based on the saved password. Thus such solutions do not provide confidentiality against device theft or malware infection. Also in the case of temporary access, the external adversary can read the user's files.

The main benefit of password-based schemes is that they could, in principle, provide good availability against the external adversary. If the external adversary gains temporary access or manages to infect the user's device with malware, the adversary can send file deletion commands to the cloud storage, but all the deleted files could be kept in the trash bin. In the case of device theft by external adversary, users could still access all their files with another device.

Unfortunately, our analysis finds that most of the commercial solutions do not, in practice, provide such protection. Some of the analyzed solutions~\cite{mega-security,woekli,zoolz} provide simple ways for users to permanently erase files from the trash bin, thus voiding this benefit.
Further, in most of the analyzed services~\cite{pcloud-encrypted-cloud-storage,woekli,spideroak-no-knowledge,tresorit-security,idrive} the external adversary, who uses an active session, is allowed to change the user's password. These services typically implement password change by having the client re-encrypt a file encryption key (that is part of the session token) with the new password and then deleting the previous copy of the file encryption key that was encrypted with the old password. Such implementation effectively locks out the legitimate user, without any possibility of recovering access to his files. Only two of the analyzed services, Sync~\cite{sync-your-privacy} and Boxcryptor~\cite{boxcryptor},  implement file versioning and password changes such that they can offer some degree of availability against the external adversary. We discuss this topic further in Appendix~\ref{sec:analysis-remarks}.

Another major problem of password-based encryption schemes is that they do not tolerate forgotten passwords. When passwords are used as an authentication mechanism, resetting them is possible. However, when passwords are used for encryption, a forgotten password means that the user will lose access to all their files.
To address this concern, some of the analyzed solutions like Sync~\cite{sync-your-privacy} and SpiderOak~\cite{spideroak-no-knowledge} provide recovery mechanisms like password hints. However, such ``helpers'' also help the could  adversary to access the encrypted files. Other services provide email-based password recovery options which in practice mean that, if such features are enabled, there is no confidentiality towards the cloud adversary at all (see Appendix~\ref{sec:analysis-remarks} for details).

\paragraph{Summary.}
We conclude that none of the analyzed solutions provides sufficient confidentiality and availability with respect to our threat model (see Table~\ref{tab:comparison}). Encryption using pseudorandom keys is secure against the cloud adversary, but it provides no confidentiality towards the external adversary. Password-based solutions provide poor confidentiality against both adversaries. Additionally, availability of such solutions is poor, as access to all files is lost in case of a forgotten password.


\section{Two-factor Encryption}
\label{sec:2fe_overview}

Motivated by this situation we propose a novel solution called \emph{Two-Factor Encryption} (2FE). Our main goal is to address the previously discussed limitations of currently available commercial services and provide a solution that is secure with respect to our extensive but realistic threat model.

\paragraph{Assumptions.} When designing our solution, we consider the system model defined in Section~\ref{sec:motivation}. The only difference is that we now assume that the user has \emph{two} devices, like a laptop and a smartphone. One of the devices, say the laptop, is called \emph{primary} and it is used to access files from the cloud. The other device, for example the smartphone, is called \emph{secondary} and its purpose is to assist the primary in file encryption and decryption. This setting is comparable to the widely adopted concept of two-factor authentication (2FA), where the primary device needs to log in to an online service, and the secondary device assists the primary in this task.

We consider the threat model defined in Section~\ref{sec:motivation} with one small refinement. We assume that the external adversary can compromise \emph{only one} of the user's devices at a time. That is, at no point in time are both user's devices controlled by the external adversary simultaneously.

\subsection{Main Approach}
\label{sec:approach}

Our analysis showed that password-based encryption cannot provide sufficient confidentiality against the cloud adversary. Therefore, we focus on solutions where user's files are encrypted using full-length pseudorandom keys. Additionally, we learned that solutions that store the entire encryption key on the user's device, when combined with convenience features such as long-lived sessions, do not provide confidentiality against the external adversary who steals the user's device. Thus, we rule out possible solutions that duplicate encryption keys to both devices. 

For these reasons, we take the well-known cryptographic technique of \emph{secret sharing} as the starting point of our solution. Our main approach is to use 2-out-of-2 secret sharing scheme and split the  encryption secret $\PRFkey$ into two shares, $\keyC$ and $\keyD$, that are stored on the primary and secondary device, respectively. Using such an approach it becomes possible to build a solution that provides strong confidentiality towards the cloud adversary and the external adversary at the same time.

\subsection{Challenges and Techniques}
\label{sec:challenges}

Just leveraging simple secret sharing is insufficient for secure encrypted cloud storage. In this section, we explain the problems of such simple solution and then describe techniques that we use to address them. By doing this, we gradually introduce the main technical ingredients of our solution and build it piece by piece. The full solution is described in Section~\ref{sec:full-protocol}.

\paragraph{Key recovery.} Simple 2-out-of-2 secret sharing has the following data availability problem. If one of the user's devices is lost or stolen by the external adversary, re-construction of the full encryption secret $\PRFkey$ becomes impossible, and the user will permanently loose access to all his files.

To address this problem we use the untrusted cloud storage as a \emph{virtual} third device. If one of the user's devices is lost or stolen, the user can recover the needed key material from the cloud storage and the other device. More precisely, we let the primary compute a 2-out-of-2 secret sharing $\keyC^{\S}$, $\keyC^{\D}$ of its key share $\keyC$ and give $\keyC^{\S}$ to the storage server and $\keyC^{\D}$ to the secondary device. Thus, if the primary device is lost, the remaining devices (the one left in the user's control plus the cloud) can recover the lost share and thus access to the user's files. The secondary device performs analogous key sharing.

\paragraph{Malicious recovery.} The above technique addresses the availability concern, but it raises a new confidentiality problem. If the external adversary gains access to one of the two devices, what prevents him from recovering the needed key shares from the cloud storage and thus accessing all user's files?

We handle this problem by requiring that the user must securely authorize every recovery operation. 
This is possible because most online service providers support some form of (out-of-band) identity verification process, as explained in Section~\ref{sec:motivation.system}. If the user loses one of his devices permanently, he can complete the identity verification process with the cloud storage. During the identity verification process the user binds the identity of the new replacement device to the recovery process, which allows the cloud storage to release the key material (e.g., $\keyC^{\S}$) to the correct replacement device. The identity verification may be a somewhat inconvenient operation, but it is performed infrequently, i.e., only when the user looses a device permanently.

\paragraph{Key share refresh.} The simple secret-sharing approach has another confidentiality problem. If the external adversary compromises one device first, and then the other device later, the adversary will \emph{eventually} learn enough information (both key shares $\keyC$ and $\keyD$) that allow him to reconstruct the full encryption secret $\PRFkey$, even though the adversary never controls both devices simultaneously.

We handle this problem with a key refresh mechanism that invalidates old key shares and computes new key shares after every device recovery operation.
To enable key refresh, we need to preserve the invariant $\keyC+\keyD=\PRFkey$ which can be achieved as follows: After recovery operation, we refresh the shares $\keyC'=\keyC+v$ and $\keyD'=\keyD-v$ for a random $v$ and delete the old shares from the user's devices. Note that $(\keyC+v)+(\keyD-v)=\PRFkey+0=\PRFkey$, so the invariant is maintained.
However, if the adversary steals $\keyC$ and then $\keyD'=\keyD-v$, he will not learn anything about $\PRFkey$ as $v$ is unknown.

\paragraph{Secure key derivation.} Our next problem is related per-file encryption key derivation, in cases where one of the devices is compromised by the external adversary. We wish to derive a per-file encryption key $k$ as $k=F(\PRFkey,s)$ where $F$ is a pseudo-random function (PRF) and $s$ some fresh randomness that ensures two files get different keys. The key derivation process should be such that $\keyC$ is not revealed to the secondary device (or $\keyD$ to the primary). Thus, we cannot simply send $\keyD$ to the primary device and let it compute $F$ on its own as that would reveal $\PRFkey$.

To address this challenge, we evaluate $F$ in an ``oblivious'' fashion.
An \emph{Oblivious PRF} (OPRF) \cite{10.1007:978-3-540-30576-7_17} is a protocol where one party holds a key $\PRFkey$ while the other holds an input $x$.
The output is the value $z=F(\PRFkey,x)$ which is given to the party who holds $x$. Such functionality is almost what we need.
Our solution needs a slightly modified oblivious protocol that goes as follows: The primary device has input $(\keyC,s)$ while the secondary device has input $(\keyD,s)$. The output $k=F({\keyC+\keyD},s)=F(\PRFkey,s)$ is given to the primary device, without revealing anything about $\keyC$ to the secondary device (or $\keyD$ to the primary device).

\paragraph{Detection of bad keys.}
File decryption only works if both the primary and the secondary device input their correct key shares in the above outlined per-file encryption key derivation process. 
If one of the devices, say the secondary, is malicious and inputs a false ``key share'' $\keyD^{*}$, the resulting file encryption key is such that it cannot be later recovered, and the availability of the encrypted file is lost.

To prevent this, we integrate simple zero-knowledge proofs to the above outlined oblivious PRF protocol that allow the primary to detect misbehaving secondary. Such functionality can be achieved by making small adjustments to the construction described in \cite{DBLP:conf/asiacrypt/JareckiKK14}.

\paragraph{Malicious primary.}
Finally, we need to consider the worst possible case for any two-factor solution---one where the \emph{primary} is malicious because the external adversary has infected it with malware or has gained temporary access to it. In such cases, the adversary will be able to request the decryption of any file they wish.

To address such cases, our solution can be complemented with notifications or prompts on the secondary device, if the user wants to enable such protections. When the primary device requests a decryption, the secondary device can either notify the user that a decryption is being made or ask the user to approve the decryption. As too frequent notifications and prompts can lead to negative effects such as user habituation, we discuss possible ways to reduce them in Section \ref{sec:reduced-prompts}.

\subsection{Functionalities}
\label{sec:functionalities}

Before describing our full solution, we define one \emph{primitive} (secret sharing) and two \emph{functionalities} (shared randomness and threshold PRF) that we, for now, assume are securely available. This approach simplifies both the presentation of our full solution (next section) and its security analysis (Section~\ref{sec:security-analysis}). In Appendix \ref{sec:2fe} we show how the assumed functionalities can be instantiated securely and efficiently such that they work together with secret sharing primitive.

\paragraph{Secret sharing.} The main technical primitive that our solution needs is secret sharing. Let $K$ be a field. A secret-sharing of $v\in K$ is a vector $(v_{0},v_{1})$ of elements picked uniformly at random from $K$ such that $v_{0}+v_{1}=v$.
Note that no information about $v$ is leaked given only $v_{0}$ or $v_{1}$; on the other hand, $v$ is uniquely determined by $(v_{0},v_{1})$.
In our solution, we will make use of the following two properties:

\begin{enumerate}
\item If each device selects a ``share'' at random themselves---call these $\keyC$ and $\keyD$---then these together uniquely define some \emph{unknown} element $\PRFkey$.
\item If $(v_{0},v_{1})$ is a secret-sharing of $v$, and $(u_{0},u_{1})$ a secret-sharing of $u$, then $(v_{0}+u_{0},v_{1}+u_{1})$ is a secret sharing of $v+u$.
\end{enumerate}

The first property allows us to set up the system without ever having to  reveal the full secret $\PRFkey$ to any of the user's devices.
The second allows us to refresh the key, by adding a secret-sharing of 0. In Appendix \ref{sec:2fe} we construct such secret sharing.

\paragraph{Shared randomness.}
The first functionality that our solution needs is one that allows two parties to generate a random value that cannot be biased. This functionality waits to receive an $\mathtt{init}$ message from both devices, and once received, outputs a uniformly random string to everyone, as shown below in Functionality \ref{func:cointoss}.

\begin{functionality}{func:cointoss}{$\SR(\lambda)$ --- Shared Randomness}
    \begin{itemize}
    \item On $\cmd{init}$ from both devices, sample $s\pick\bits^{\lambda}$ at random and output $\cmda{seed}{s}$ to both devices.
    \end{itemize}
\end{functionality}

Instantiating the above $\SR$ protocol is straightforward using commitments; a concrete instantiation and its security argument are provided in Appendix \ref{sec:2fe}.

\paragraph{Threshold PRF.}
The second functionality that our solution needs is one that allows two devices to evaluate a pseudo-random function to obtain a key suitable for encryption.
More precisely, this functionality will evaluate $F(\keyC+\keyD,x)$ based on inputs $x$ and $\keyC$ from the primary and input $\keyD$ from the secondary without revealing $\keyD$ to the primary or $\keyC$ to the secondary.
Further, we require that this functionality can detect whether the secondary device inputs an incorrect key.
This is done by having the primary input a ``public key'' $\pk$, and checking if $\pk=\mathsf{Gen}(\keyD)$ where $\keyD$ is the key input by the secondary, and $\mathsf{Gen}$ is a deterministic key generation procedure.
This protocol is described below in Functionality \ref{func:tprf}.

\begin{functionality}{func:tprf}{$\tPRF(\lambda)$ --- Threshold PRF}
  Let $F:K\times\bits^{\ell}\to\bits^{\lambda}$ be a PRF with keyspace $K$.
  \begin{itemize}
    \item On $\cmda{init}{x,\keyC,\pk}$ from the primary device, store $(x,\pk)$ and wait for a message from the secondary device.
    \item On $\cmda{init}{x',\keyD}$ from the secondary device, then:
      \begin{itemize}
      \item If $x\neq x'$ or $\pk\neq\mathsf{Gen}(\keyD)$ set $k=\bot$.
      \item Otherwise, pick at random $k=F(\keyC+\keyD,x)$.
      \end{itemize}
    \item Send $\cmda{key}{k}$ to the primary device.
    \end{itemize}
\end{functionality}

In Appendix \ref{sec:2fe} we show how to instantiate such $\tPRF$ from the oblivious PRF protocol in \cite{DBLP:conf/asiacrypt/JareckiKK14} and provide a security argument for our construction.

\subsection{Full Solution}
\label{sec:full-protocol}

We are now ready to present our solution in detail.

\paragraph{Session management.}
We manage sessions in a straightforward way.
During system initialization, the user's two devices perform a pairing after which the user authenticates towards the storage server using using an authentication credential such as password.
The user can enable convenience features such as long-lived sessions that allow  encryption and decryption of files, without having to perform authentication for every operation separately. One device can be used to invalidate the session of the other device (e.g., in case of lost or stolen device).

\paragraph{Enrollment.}
Figure~\ref{fig:twofe-init} shows how devices are enrolled. In step 1, each of the user's devices generate their shares of the master encryption key, as well as a secret sharing of their key-shares.
That is, each device picks a random key, e.g., $\keyC$ for the primary and computes as well a sharing $(\keyC^{\D},\keyC^{\S})$ of $\keyC$.
Step 2 and 3 are executed in parallel and involves the devices sending the shares of their encryption keys to the other device and the storage server.
In step 2, the secondary device additionally sends the public key of its key share $\pk=\mathsf{Gen}(\keyD)$ to the primary, which the latter will use as input to the threshold PRF functionality.

\paragraph{Encryption.}
Figure~\ref{fig:twofe-enc} shows how \name encrypts a file $m$.
Given $m$, the primary device generates a tag $t$ which will be used to later identify the ciphertext during decryption (step 2).
The only requirements we put on $t$ is to be unique (in particular, the system will use $t$ to identify files) and that $t$ does not reveal information about files.
One way to handle these tags would be for each device to store a mapping $\mathtt{filename}\to t$ and then pick $t$ at random. This mapping could be stored under a static key on the cloud as well (i.e., a key which does not require two factor decryption).

Next, in step 3, the primary and secondary invoke the shared randomness $\SR$ functionality to generate a fresh encryption seed.
The primary device then sends $t$ to the secondary device (step 4) and the two devices then invoke the $\tPRF$ functionality to derive an encryption key (step 5).
Finally, using $k$, the primary device can encrypt $m$ (step 6) and then upload the ciphertext, as well as $t$ and $s$ which are needed for decryption, to the cloud (step 7).

\begin{figure}[t]
  \centering
  \subfloat {
    \scalebox{0.73}{
      \tikzset{
	partial ellipse/.style args={#1:#2:#3}{
		insert path={+ (#1:#3) arc (#1:#2:#3)}
	}
}

\tikzset{
	ncbar angle/.initial=90,
	ncbar/.style={
		to path=(\tikztostart)
		-- ($(\tikztostart)!#1!\pgfkeysvalueof{/tikz/ncbar angle}:(\tikztotarget)$)
		-- ($(\tikztotarget)!($(\tikztostart)!#1!\pgfkeysvalueof{/tikz/ncbar angle}:(\tikztotarget)$)!\pgfkeysvalueof{/tikz/ncbar angle}:(\tikztostart)$)
		-- (\tikztotarget)
	},
	ncbar/.default=0.5cm,
}

\begin{tikzpicture}[every node/.style={transform shape},apply style/.code={\tikzset{#1}},>=stealth']
\def\hordist{3cm}
\def\hordistw{3.5cm}
\def\slope{0.3}
\def\spacing{0}

\node[icon] (user) {\pgfuseimage{user}};
\node[txt,above=0cm of user] (t1) {User};
\node[right=\hordist of user.center,anchor=center,icon] (device) {\pgfuseimage{device}};
\node[txt] at (device |- t1) {Secondary};
\node[right=\hordistw of device.center, anchor=center,icon] (client) {\pgfuseimage{client}} {};
\node[txt] (client-text) at (client |- t1) {Primary};
\node[right=\hordistw of client.center, anchor=center,icon] (cloud) {\pgfuseimage{cloud}};
\node[txt] at (cloud |- t1) {Cloud};

\def\length{5.2}

\draw[-] ($(user)-(0, .8)$) -- ($(user)-(0,\length)$);
\draw[-] ($(device)-(0, .8)$) -- ($(device)-(0,\length)$);
\draw[-] ($(client)-(0, .8)$) -- ($(client)-(0,\length)$);
\draw[-] ($(cloud)-(0, .8)$) -- ($(cloud)-(0,\length)$);

\node at ($(device)-(-.2,1.2)$) {1};
\node at ($(client)-(.2,1.2)$) {1};
\node at ($(client)-(.2,3)$) {2};
\node at ($(device)-(.2,4)$) {3};

\def\baseline{1.2}
\def\slope{0.2}

\draw[->,color=black!100] ($(client) - (0,\baseline)$) to [out=0,in=0,looseness=2] node[right,text width=\hordist*.6,align=center]{Generate\\$\keyC,\keyC^{\S},\keyC^{\D}$} ($(client) - (0,\baseline+1)$);

\draw[->,color=black!100] ($(device) - (0,\baseline)$) to [out=180,in=180,looseness=2] node[left,text width=\hordist*.7,align=center]{Generate\\$\keyD,\keyD^{\S},\keyD^{\C},\pk$} ($(device) - (0,\baseline+1)$);

\draw[->,color=black!100] ($(client) - (0,\baseline+1.6)$) -- node[above,align=center] {$\keyC^{\D}$} ($(device) - (0,\baseline+1.8)$);
\draw[->,color=black!100] ($(client) - (0,\baseline+1.6)$) -- node[above,align=center] {$\keyC^{\S}$} ($(cloud) - (0,\baseline+1.8)$);

\draw[->,color=black!100] ($(device)-(0,\baseline+2.5)$) -- node[above,align=center] {$\keyD^{\C},\pk$} ($(client)-(0,\baseline+2.9)$);

\draw[->,color=black!100] ($(device)-(0,\baseline+3)$) -- node[above,align=center,pos=0.75] {$\keyD^{\S}$} ($(cloud)-(0,\baseline+3.8)$);

\end{tikzpicture}
    }
  }
  \caption{System enrollment.}
  \label{fig:twofe-init}
\end{figure}
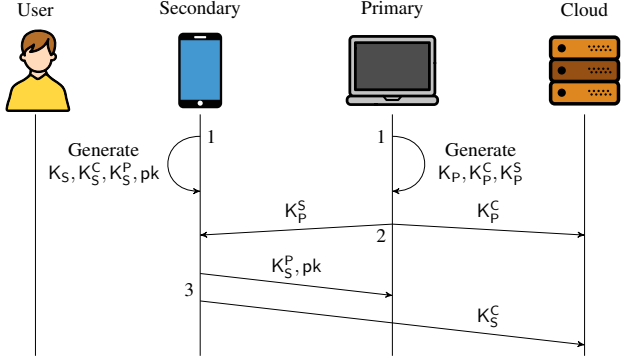

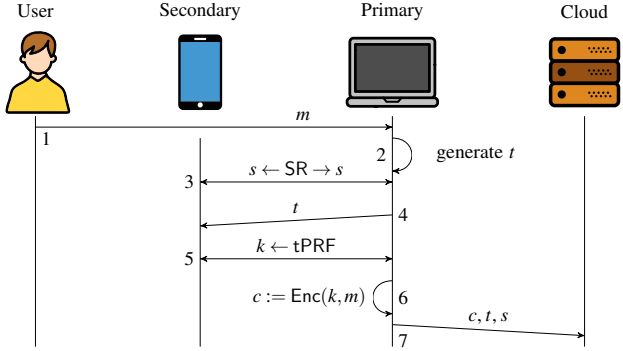
\begin{figure}[tb]
  \centering
  \subfloat {
    \scalebox{0.73}{
      \tikzset{
	partial ellipse/.style args={#1:#2:#3}{
		insert path={+ (#1:#3) arc (#1:#2:#3)}
	}
}

\tikzset{
	ncbar angle/.initial=90,
	ncbar/.style={
		to path=(\tikztostart)
		-- ($(\tikztostart)!#1!\pgfkeysvalueof{/tikz/ncbar angle}:(\tikztotarget)$)
		-- ($(\tikztotarget)!($(\tikztostart)!#1!\pgfkeysvalueof{/tikz/ncbar angle}:(\tikztotarget)$)!\pgfkeysvalueof{/tikz/ncbar angle}:(\tikztostart)$)
		-- (\tikztotarget)
	},
	ncbar/.default=0.5cm,
}

\begin{tikzpicture}[every node/.style={transform shape},apply style/.code={\tikzset{#1}},>=stealth']
\def\hordist{3cm}
\def\hordistw{3.5cm}
\def\slope{0.3}
\def\spacing{0}

\node[icon] (user) {\pgfuseimage{user}};
\node[txt,above=0cm of user] (t1) {User};
\node[right=\hordist of user.center,anchor=center,icon] (device) {\pgfuseimage{device}};
\node[txt] at (device |- t1) {Secondary};
\node[right=\hordistw of device.center, anchor=center,icon] (client) {\pgfuseimage{client}} {};
\node[txt] (client-text) at (client |- t1) {Primary};
\node[right=\hordistw of client.center, anchor=center,icon] (cloud) {\pgfuseimage{cloud}};
\node[txt] at (cloud |- t1) {Cloud};

\def\length{5}

\draw[-] ($(user)-(0, .8)$) -- ($(user)-(0,\length)$);
\draw[-] ($(device)-(0, 1.2)$) -- ($(device)-(0,\length)$);
\draw[-] ($(client)-(0, .8)$) -- ($(client)-(0,\length)$);
\draw[-] ($(cloud)-(0, .8)$) -- ($(cloud)-(0,\length)$);

\node at ($(user)-(-.2,1.2)$) {1};
\node at ($(client)-(.2,1.5)$) {2};
\node at ($(device)-(.2,2)$) {3};
\node at ($(client)-(-.2,2.6)$) {4};
\node at ($(device)-(.2,3.4)$) {5};
\node at ($(client)-(-.2,4.1)$) {6};
\node at ($(client)-(-.2,4.9)$) {7};

\def\baseline{1}
\def\slope{0.2}

\draw[->,color=black!100] ($(user) - (0,\baseline)$) -- node[above,align=center,pos=.75] {$m$} ($(client) - (0,\baseline)$);

\draw[->,color=black!100] ($(client) - (0, \baseline+0.2)$) to [out=0,in=0,looseness=2] node[right,text width=\hordist*.7,align=center,fill=white!0] {generate $t$} ($(client) - (0,\baseline+0.8)$);

\draw[<->,color=black!100] ($(client) - (0,\baseline+1)$) -- node[above,align=center] {$s\leftarrow\SR\rightarrow s$} ($(device) - (0,\baseline+1)$);

\draw[->,color=black!100] ($(client) - (0,\baseline+1.6)$) -- node[above,align=center] {$t$} ($(device) - (0,\baseline+1.8)$);

\draw[<->,color=black!100] ($(client) - (0,\baseline+2.4)$) -- node[above,align=center] {$k\leftarrow\mathsf{tPRF}$} ($(device) - (0,\baseline+2.4)$);

\draw[->,color=black!100] ($(client) - (0, \baseline+2.8)$) to [out=180,in=180,looseness=2]
node[left,text width=\hordist*.7,align=center,fill=white!0] {$c:=\Enc(k,m)$} ($(client) - (0,\baseline+2.8+.6)$);

\draw[->,color=black!100] ($(client) - (0,\baseline+3.6)$) -- node[above,align=center] {$c$, $t$, $s$} ($(cloud) - (0,\baseline+3.8)$);




\end{tikzpicture}
    }
  }
  \caption{File encryption.}
  \label{fig:twofe-enc}
\end{figure}

\paragraph{Decryption.}
Figure~\ref{fig:twofe-dec} presents file decryption.
In step 1, the user identifies which file they want to decrypt by inputting $t$ to the primary device.
In step 2 and 3, the primary device retrieves the ciphertext $c$ and the encryption seed $s$ by sending $t$ to the storage server.
The tag and seed are then forwarded to the secondary device.
Depending on its configuration, \name can transparently accept the decryption request, or, if enabled by the user, show a decryption notification to the user or ask the user to confirm the decryption of the specified file and tag (steps 5 and 6).
After the operation is approved, in step 7, both devices invoke the $\tPRF$ functionality to recover the decryption key.
Finally, in step 8, the primary device decrypts $m$, and returns it to the user in step 9.

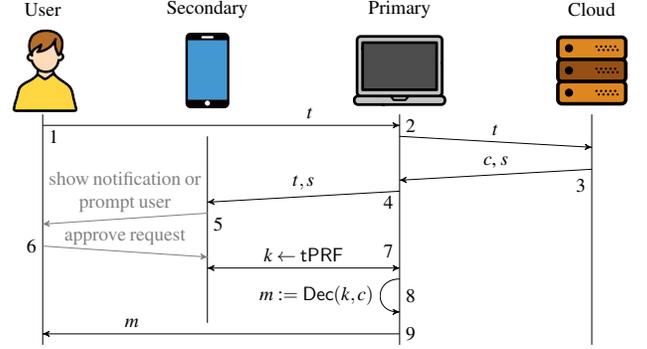
\begin{figure}[t]
  \centering
  \subfloat {
    \scalebox{0.73}{
      \tikzset{
	partial ellipse/.style args={#1:#2:#3}{
		insert path={+ (#1:#3) arc (#1:#2:#3)}
	}
}

\tikzset{
	ncbar angle/.initial=90,
	ncbar/.style={
		to path=(\tikztostart)
		-- ($(\tikztostart)!#1!\pgfkeysvalueof{/tikz/ncbar angle}:(\tikztotarget)$)
		-- ($(\tikztotarget)!($(\tikztostart)!#1!\pgfkeysvalueof{/tikz/ncbar angle}:(\tikztotarget)$)!\pgfkeysvalueof{/tikz/ncbar angle}:(\tikztostart)$)
		-- (\tikztotarget)
	},
	ncbar/.default=0.5cm,
}

\begin{tikzpicture}[every node/.style={transform shape},apply style/.code={\tikzset{#1}},>=stealth']
\def\hordist{3cm}
\def\hordistw{3.5cm}
\def\slope{0.3}
\def\spacing{0}

\node[icon] (user) {\pgfuseimage{user}};
\node[txt,above=0cm of user] (t1) {User};
\node[right=\hordist of user.center,anchor=center,icon] (device) {\pgfuseimage{device}};
\node[txt] at (device |- t1) {Secondary};
\node[right=\hordistw of device.center, anchor=center,icon] (client) {\pgfuseimage{client}} {};
\node[txt] (client-text) at (client |- t1) {Primary};
\node[right=\hordistw of client.center, anchor=center,icon] (cloud) {\pgfuseimage{cloud}};
\node[txt] at (cloud |- t1) {Cloud};

\def\length{5}

\draw[-] ($(user)-(0, .8)$) -- ($(user)-(0,\length)$);
\draw[-] ($(device)-(0, 1.2)$) -- ($(device)-(0,\length-0.4)$);
\draw[-] ($(client)-(0, .8)$) -- ($(client)-(0,\length)$);
\draw[-] ($(cloud)-(0, .8)$) -- ($(cloud)-(0,\length)$);

\def\baseline{1}
\def\slope{0.2}

\node at ($(user)-(-.2,1.2)$) {1};
\node at ($(client)-(-.2,1)$) {2};
\node at ($(cloud)-(.2,2.1)$) {3};
\node at ($(client)-(.2,2.4)$) {4};
\node at ($(device)-(-.2,2.8)$) {5};
\node at ($(user)-(.2,3.2)$) {6};
\node at ($(client)-(.2,3.3)$) {7};
\node at ($(client)-(-.2,4.1)$) {8};
\node at ($(client)-(-.2,4.8)$) {9};

\draw[->,color=black!100] ($(user) - (0,\baseline)$) -- node[above,align=center,pos=.75] {$t$} ($(client) - (0,\baseline)$);

\draw[->,color=black!100] ($(client) - (0,\baseline+0.2)$) -- node[above,align=center] {$t$} ($(cloud) - (0,\baseline+.4)$);
\draw[->,color=black!100] ($(cloud) - (0,\baseline+0.8)$) -- node[above,align=center] {$c$, $s$} ($(client) - (0,\baseline+1)$);

\draw[->,color=black!100] ($(client) - (0,\baseline+1.2)$) -- node[above,align=center] {$t, s$} ($(device) - (0,\baseline+1.4)$);

\draw[->,color=gray!100] ($(device) - (0,\baseline+1.6)$) -- node[above,align=center] {show notification or \\prompt user} ($(user) - (0,\baseline+1.8)$);
\draw[->,color=gray!100] ($(user) - (0,\baseline+2.2)$) -- node[above,align=center] {approve request} ($(device) - (0,\baseline+2.4)$);

\draw[<->,color=black!100] ($(client) - (0,\baseline+2.6)$) -- node[above,align=center] {$k\leftarrow\mathsf{tPRF}$} ($(device) - (0,\baseline+2.6)$);

\draw[->,color=black!100] ($(client) - (0, \baseline+2.8)$) to [out=180,in=180,looseness=2]
node[left,text width=\hordist*.7,align=center,fill=white!0] {$m:=\Dec(k,c)$} ($(client) - (0,\baseline+2.8+.6)$);

\draw[->,color=black!100] ($(client) - (0,\baseline+3.8)$) -- node[above,align=center,pos=.75] {$m$} ($(user) - (0,\baseline+3.8)$);



\end{tikzpicture}
    }
  }
  \caption{File decryption (notification/prompt optional).}
  \label{fig:twofe-dec}
\end{figure}

\paragraph{Migration.}
Figure~\ref{fig:twofe-rec2} shows migration of the secondary device (primary device migration is analogous).
In step 1, the primary device sends a recover request to the cloud.
The cloud then sends an authorization message to the user's \emph{old} device, which then prompts the user for authorization (steps 3 and 4).
This step ensures that the user actually wishes to replace the old device, as well as gives the user a chance to input binding information about the new device.
In step 5, the authorization is sent back to the server which then (using information in the authorization request) sends its key share to the user's \emph{new} device (step 6). Finally, the primary device sends its key share as well (step 7) and the primary and new secondary device update their keys (step 8), which works as follows:
Primary device generates a secret-sharing $(v_{\C},v_{\D})$ of $0$, sets $\keyC:=\keyC+v_{\C}$ and sends $v_{\D}$ to the secondary, who sets $\keyD:=\keyD+v_{\D}$.
Next, both devices act as in system enrollment (Figure \ref{fig:twofe-init}), except that no new $\keyC$ and $\keyD$ are generated.

\paragraph{Recovery.}
Figure~\ref{fig:twofe-rec} shows device recovery for cases where the old device is lost or stolen. Similar to the above describe device migration, the process starts with a recover request from the primary to the cloud (step 1).
When the user's old device does not respond (in step 2), cloud falls back to an out-of-band identity verification (step 3).
If the identity verification is successful, recovery proceeds as before (using device binding information provided during the identity verification):
the cloud sends its share to the new device in step 4, using information in the authentication response.
The primary device sends its share (step 5) and the two new devices update their keys as explained above (step 6).

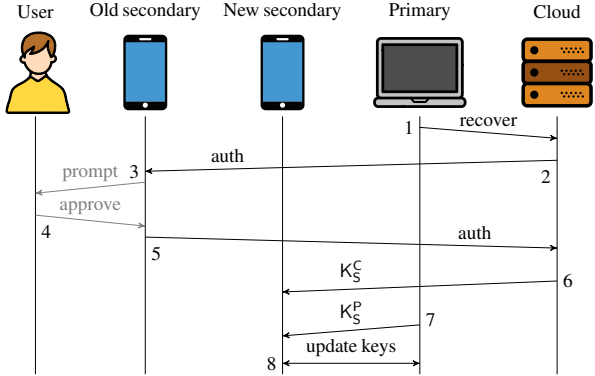
\begin{figure}[t]
  \centering
  \subfloat {
    \scalebox{0.73}{
      \tikzset{
	partial ellipse/.style args={#1:#2:#3}{
		insert path={+ (#1:#3) arc (#1:#2:#3)}
	}
}

\tikzset{
	ncbar angle/.initial=90,
	ncbar/.style={
		to path=(\tikztostart)
		-- ($(\tikztostart)!#1!\pgfkeysvalueof{/tikz/ncbar angle}:(\tikztotarget)$)
		-- ($(\tikztotarget)!($(\tikztostart)!#1!\pgfkeysvalueof{/tikz/ncbar angle}:(\tikztotarget)$)!\pgfkeysvalueof{/tikz/ncbar angle}:(\tikztostart)$)
		-- (\tikztotarget)
	},
	ncbar/.default=0.5cm,
}

\begin{tikzpicture}[every node/.style={transform shape},apply style/.code={\tikzset{#1}},>=stealth']
\def\hordist{2cm}
\def\hordistw{2.5cm}
\def\slope{0.3}
\def\spacing{0}

\node[icon] (user) {\pgfuseimage{user}};
\node[txt,above=0cm of user] (t1) {User};

\node[right=\hordist of user.center,anchor=center,icon] (device) {\pgfuseimage{device}};
\node[txt] at (device |- t1) {Old secondary};

\node[right=\hordistw of device.center, anchor=center,icon] (device2) {\pgfuseimage{device}};
\node[txt] at (device2 |- t1) {New secondary};

\node[right=\hordistw of device2.center, anchor=center,icon] (client) {\pgfuseimage{client}} {};
\node[txt] (client-text) at (client |- t1) {Primary};

\node[right=\hordistw of client.center, anchor=center,icon] (cloud) {\pgfuseimage{cloud}};
\node[txt] at (cloud |- t1) {Cloud};


\def\length{5.5}

\draw[-] ($(user)-(0, .8)$) -- ($(user)-(0,\length)$);
\draw[-] ($(device)-(0, .8)$) -- ($(device)-(0,\length)$);
\draw[-] ($(device2)-(0, .8)$) -- ($(device2)-(0,\length)$);
\draw[-] ($(client)-(0, .8)$) -- ($(client)-(0,\length)$);
\draw[-] ($(cloud)-(0, .8)$) -- ($(cloud)-(0,\length)$);

\node at ($(client)-(.2,1)$) {1};
\node at ($(cloud)-(.2,1.9)$) {2};
\node at ($(device)-(.2,1.8)$) {3};
\node at ($(user)-(-.2,2.9)$) {4};
\node at ($(device)-(-.2,3.3)$) {5};
\node at ($(cloud)-(-.2,3.8)$) {6};
\node at ($(client)-(-.2,4.5)$) {7};
\node at ($(device2)-(.2,5.3)$) {8};

\def\baseline{1}
\def\slope{0.2}

\draw[->,color=black!100] ($(client) - (0,\baseline)$) -- node[above,align=center] {recover} ($(cloud) - (0,\baseline+.2)$);

\draw[->,color=black!100] ($(cloud) - (0,\baseline+.6)$) -- node[above,align=center,pos=0.8] {auth} ($(device) - (0,\baseline+.8)$);

\draw[->,color=gray!100] ($(device) - (0,\baseline+1)$) -- node[above,align=center] {prompt} ($(user) - (0,\baseline+1.2)$);

\draw[->,color=gray!100] ($(user) - (0,\baseline+1.6)$) -- node[above,align=center] {approve} ($(device) - (0,\baseline+1.8)$);

\draw[->,color=black!100] ($(device) - (0,\baseline+2)$) -- node[above,align=center,pos=0.8] {auth} ($(cloud) - (0,\baseline+2.2)$);

\draw[->,color=black!100] ($(cloud) - (0,\baseline+2.8)$) -- node[above,align=center,pos=0.75] {$\keyD^{\S}$} ($(device2) - (0,\baseline+3)$);

\draw[->,color=black!100] ($(client) - (0,\baseline+3.6)$) -- node[above,align=center] {$\keyD^{\C}$} ($(device2) - (0,\baseline+3.8)$);

\draw[<->,color=black!100] ($(client) - (0,\baseline+4.3)$) -- node[above,align=center] {update keys} ($(device2) - (0,\baseline+4.3)$);

\end{tikzpicture}
    }
  }
  \caption{Device migration. The authorization in step 3 can also be sent directly from the primary device.}
  \label{fig:twofe-rec2}
\end{figure}


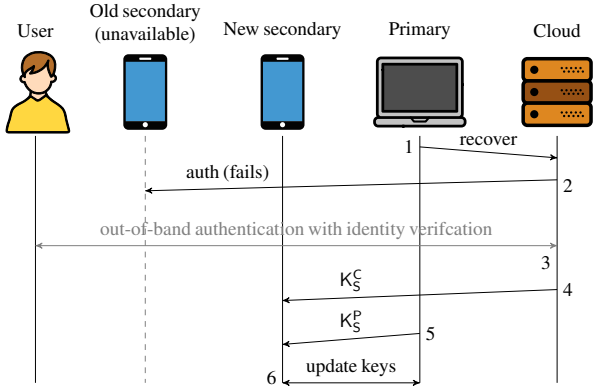
\begin{figure}[tb]
  \centering
  \subfloat {
    \scalebox{0.73}{
      \tikzset{
	partial ellipse/.style args={#1:#2:#3}{
		insert path={+ (#1:#3) arc (#1:#2:#3)}
	}
}

\tikzset{
	ncbar angle/.initial=90,
	ncbar/.style={
		to path=(\tikztostart)
		-- ($(\tikztostart)!#1!\pgfkeysvalueof{/tikz/ncbar angle}:(\tikztotarget)$)
		-- ($(\tikztotarget)!($(\tikztostart)!#1!\pgfkeysvalueof{/tikz/ncbar angle}:(\tikztotarget)$)!\pgfkeysvalueof{/tikz/ncbar angle}:(\tikztostart)$)
		-- (\tikztotarget)
	},
	ncbar/.default=0.5cm,
}

\begin{tikzpicture}[every node/.style={transform shape},apply style/.code={\tikzset{#1}},>=stealth']
\def\hordist{2cm}
\def\hordistw{2.5cm}
\def\slope{0.3}
\def\spacing{0}

\node[icon] (user) {\pgfuseimage{user}};
\node[txt,above=0cm of user] (t1) {User};

\node[right=\hordist of user.center,anchor=center,icon] (device) {\pgfuseimage{device}};
\node[txt,above=-1em of device] at (device |- t1) {Old secondary\\(unavailable)};

\node[right=\hordistw of device.center, anchor=center,icon] (device2) {\pgfuseimage{device}};
\node[txt] at (device2 |- t1) {New secondary};

\node[right=\hordistw of device2.center, anchor=center,icon] (client) {\pgfuseimage{client}} {};
\node[txt] (client-text) at (client |- t1) {Primary};

\node[right=\hordistw of client.center, anchor=center,icon] (cloud) {\pgfuseimage{cloud}};
\node[txt] at (cloud |- t1) {Cloud};


\def\length{5.3}

\draw[-] ($(user)-(0, .8)$) -- ($(user)-(0,\length)$);
\draw[-,dashed,gray!100] ($(device)-(0, .8)$) -- ($(device)-(0,\length)$);
\draw[-] ($(device2)-(0, .8)$) -- ($(device2)-(0,\length)$);
\draw[-] ($(client)-(0, .8)$) -- ($(client)-(0,\length)$);
\draw[-] ($(cloud)-(0, .8)$) -- ($(cloud)-(0,\length)$);

\node at ($(client)-(.2,1)$) {1};
\node at ($(cloud)-(-.2,1.7)$) {2};
\node at ($(cloud)-(.2,3.1)$) {3};
\node at ($(cloud)-(-.2,3.6)$) {4};
\node at ($(client)-(-.2,4.4)$) {5};
\node at ($(device2)-(.2,5.2)$) {6};

\def\baseline{1}
\def\slope{0.2}

\draw[->,color=black!100] ($(client) - (0,\baseline)$) -- node[above,align=center] {recover} ($(cloud) - (0,\baseline+.2)$);

\draw[->,color=black!100] ($(cloud) - (0,\baseline+.6)$) -- node[above,align=center,pos=0.8] {auth (fails)} ($(device) - (0,\baseline+.8)$);


\draw[<->,color=gray!100] ($(user) - (0,\baseline+1.8)$) -- node[above,align=center,pos=0.5] {out-of-band authentication with identity verifcation} ($(cloud) - (0,\baseline+1.8)$);

\draw[->,color=black!100] ($(cloud) - (0,\baseline+2.6)$) -- node[above,align=center,pos=.75] {$\keyD^{\S}$} ($(device2) - (0,\baseline+2.8)$);

\draw[->,color=black!100] ($(client) - (0,\baseline+3.4)$) -- node[above,align=center] {$\keyD^{\C}$} ($(device2) - (0,\baseline+3.6)$);

\draw[<->,color=black!100] ($(client) - (0,\baseline+4.3)$) -- node[above,align=center] {update keys} ($(device2) - (0,\baseline+4.3)$);



\end{tikzpicture}
    }
  }
  \caption{Device recovery.}
  \label{fig:twofe-rec}
\end{figure}


\section{Security Analysis}
\label{sec:security-analysis}

In this section we analyze the security of \name. We consider the simpler case of the cloud adversary first and then the more complicated case of external adversary for the secondary and the primary device, respectively. Finally, we analyze the security of the recovery mechanism.

\subsection{Cloud Storage Adversary}
\label{sec:security.cloud}
Against the cloud adversary we only consider confidentiality of the user's files, showing that (1) the cloud knows nothing about the user's keys, and (2) that keys derived with $\tPRF$ are sufficiently random, i.e., suitable for file encryption.
For (1), notice that the cloud holds $\keyD^{\S}$ and $\keyC^{\S}$.
However, none of these reveal any information about $\keyD$, $\keyC$ or $\PRFkey$, so no sensitive information about encryption keys is revealed to the cloud.
For (2), we can rely on our argument regarding the first point.
That is, since the cloud knows nothing about $\PRFkey$, it also has no information about keys outputted by $\tPRF$.

\subsection{External Adversary: Secondary Device}
\label{sec:security.secondary}

Regarding the external adversary that attacks the secondary device we consider both confidentiality and availability of the user's files.
For the arguments that follow, we rely on the following observations:
(1) the secondary device cannot influence the randomness of the key-seed derived using the $\SR$ functionality;
(2) the secondary device cannot provide malformed input to the $\tPRF$ functionality without being detected;
(3) the secondary device has no information about $\PRFkey$ due to our choice of secret-sharing scheme;
and (4) the secondary device cannot initiate file decryption or encryption.

\paragraph{Stolen device.}
In this case, the external adversary can access all data at rest on the stolen device.
From (3) we know that no information is revealed about $\PRFkey$.
Additionally, the adversary cannot successfully complete the identity verification protocol and so cannot use the stolen device to ``recover'' the primary device and gain its key share.

\paragraph{Temporary access.}
Here, the user may perform a number of encryption or decryption queries while the adversary has temporary access to the secondary device.
As above, no information is leaked about $\PRFkey$, and it is not possible to start the recovery process, as the user has to approve it.

\paragraph{Malware.}
Malware on the secondary device is similar to temporary access, with a further threat: the malware could tamper with the protocols being run, for example by sending malformed messages.
However, even with such capabilities, the adversary cannot violate either confidentiality or availability.
From (1) we know that the adversary cannot influence the entropy of keys during encryption, and so confidentiality will be upheld (as file encryption will be sufficiently strong).
Furthermore, from (2) the adversary cannot make the derived keys malformed, preserving availability.

\subsection{External Adversary: Primary Device}
\label{sec:security.primary}

Regarding the external adversary that attacks the primary device we use the following observations:
(1) the primary device cannot influence the key-seed, from the definition of the $\SR$ functionality;
(2) the primary device has no information about $\PRFkey$ due to our choice of secret-sharing scheme;
and (3) if user prompts are enabled, the user has the final say in whether decryption is run (or if notifications are enabled the user is informed of each decryption).

\paragraph{Stolen device.}
Similarly to stealing the secondary device, the adversary cannot break confidentiality as it only has one share (1).
Likewise, the adversary cannot start a recovery as it cannot successfully complete the identity verification protocol.

\paragraph{Temporary access and malware.}
If no user prompts (or notifications) have been enabled, then the adversary can break confidentiality as the primary device can request decryptions.
On the other hand, with user prompts, the user would have to approve queries on the secondary device, and so confidentiality is preserved. Notifications allow the user to detect unauthorized decryptions and take corrective action (e.g., clean his device, disable active sessions and change his password).
An external adversary with temporary access cannot encrypt files that cannot later be decrypted, and in this case availability is preserved. File encryption during malware infection obviously cannot provide data availability, as the adversary can replace file contents with arbitrary data (recall Section~\ref{sec:analysis:properties}).

\subsection{Recovery Security}
\label{sec:security-recovery}
Finally, we analyze the security of the device recovery and migration protocols when the adversary has control over one of the user's devices. 
Regardless of whether an adversary attacks one of the user's devices or the storage server, the same idea holds:
Recall that the master key $\PRFkey$ can be written $\PRFkey = \keyC + \keyD = (\keyC^{\S} + \keyC^{\D}) + (\keyD^{\S} + \keyD^{\C})$.
It is quite easy to see that an attacker never sees enough key material in order to determine $\PRFkey$.
Indeed, the above equality can be alternatively written as one of the following three equations depending on which device (or the cloud) that is compromised:
\begin{align*}
  \PRFkey &= (\keyC + \keyD^{\C}) + \keyC^{\S},\\
          &= (\keyD + \keyC^{\D}) + \keyC^{\S},\\
          &= (\keyC^{\S} + \keyD^{\S}) + \keyC^{\D} + \keyD^{\C}.
\end{align*}
In each case, the attacker learns only the information in the parenthesis.
The remaining values are, by definition, uniformly at random and thus reveal nothing about $\PRFkey$.
It only remains now to argue that the protocols for migration and recovery themselves do not lead to an attack where more than one key-share can be learned, which we do next. In the following, $\C$ stands for primary, $\D$ for secondary and $\S$ for cloud storage.

\paragraph{Adversary controls $\C$, tries to recover $\D$.}
This matches the case described in Section \ref{sec:security.primary}.
In case of device migration, the old $\D$ is still active (will respond to a query by $\S$) and so the user will be notified that a recovery has been started.
Thus, the user will see that something is amiss and can take appropriate actions.
On the other hand, if $\D$ is not available, the cloud will behave according to Figure \ref{fig:twofe-rec}.
In this case, and because we assume that the adversary cannot successfully complete the out-of-band identity verification, it will not be able to get the cloud storage to send its key share to a device controlled by the adversary.

\paragraph{Adversary controls $\D$, tries to recover $\C$.}
This proceeds exactly as above: if $\C$ is responding, the user will detect (and abort) the recovery; if $\C$ is not responding, as the adversary cannot complete the out-of-band authentication, no identifying information can be supplied by the adversary.

\paragraph{Adversary controls $\C$, tries to recover $\C$.}
Interestingly, this ``attack'' can in fact be carried out by the adversary.
Since the server will contact the old device (in this case $\C$, which is controlled by the adversary), the adversary can supply information for a device it controls.
However, no \emph{new} information is made available to the adversary.
Indeed, all it learns are the key information stored on $\C$ which it already have access to.

\paragraph{Adversary controls $\D$, tries to recover $\D$.}
Same argument applies here: no new information is revealed to the adversary.
In effect, all that happens is that the system migrates to another device controlled by the adversary.
It should be noted, however, that this attack could be used by the adversary to escalate temporary access of the device to permanent access---at least until the user runs the recover protocol.

\subsection{Usability Considerations}
\label{sec:reduced-prompts}

Our solution can be deployed in two ways. The simplest option is without user prompts. In this case, the user needs to \emph{own} and enroll two devices, but the user does not have to \emph{do} anything on the secondary device during file encryption/decryption. Thus, the user experience is unchanged from standard cloud storage (the only difference is a minor key derivation delay which we evaluate in Section~\ref{sec:evaluation}). 

The second way to deploy \name is with decryption prompts on the secondary device for added protection. We acknowledged that frequent user prompts and notifications can have negative effects such as habituation. One way to reduce the number of prompts and notifications is to let the user to set different security levels for different files.
For example, the user could set the security of their calendar to be ``low'', meaning that decryption of the calendar would happen without any prompt; while the user's tax returns could be set to ``high'', meaning the decryption would need to prompt the user. Such policies can be applied on a per-folder basis.
Another possible option is to allow the user to approve the first decryption of a file; then, all subsequent decryptions within a user-defined time window would automatically get approved.

\subsection{Analysis Summary}
The last two rows of Table~\ref{tab:comparison} summarize our  analysis. When user prompts are enabled, \name provides all of our security properties (data availability in the case of malware is marked as partially-provided, because availability holds only for files that were encrypted and uploaded \emph{before} the primary was infected).
When user prompts or notifications are not used, we also mark confidentiality for temporary access and malware as partially-provided properties, because confidentiality is preserved in the case of compromised secondary, but not primary. Compared to simple key sharing, key duplication, and all the analyzed commercial services, \name provides better confidentiality and availability.


\section{Performance Evaluation}
\label{sec:evaluation}

We implemented the primary device as a Python application, and the secondary device as an Android application.
The key derivation protocol implementation relies on Elliptic-Curve cryptography --- we use \texttt{fastecdsa}~\cite{fastecdsa} for the primary device, and \texttt{Spongy Castle}~\cite{spongyCastle} for the secondary device.
We use $\mathsf{SHA256}$ as the hash function. 
Our system abstracts the storage backend and allows to use any existing provider as an untrusted cloud storage. For our prototype we use Dropbox for storing and retrieving ciphertext.
Secondary and primary devices communicate with each other via TCP sockets. 
For our prototype, the primary laptop device withe the Python app implements the server side of the TCP sockets. After creation, the primary app sends the connection parameters (IP address, port number) to the Android app in a notification message~\cite{firebaseCloudMsg} which enables seamless connection between the two devices.

\subsection{Experiments}

In a typical encrypted cloud storage scheme, the client generates a symmetric key, encrypts the file and sends the ciphertext to the server.
The overhead of \name compared to such standard single-device encryption system is the interactive key derivation process. 
Essentially, our solution replaces the local key derivation with an interactive protocol that is run between the two user's devices. To evaluate performance overhead of \name, we measure the delay of the interactive protocol that is executed to derive the needed keys.
We measure the time it takes to compute a file key for encryption and decryption separately\footnote{Decryption requires less interaction because the seed is provided by the primary device, instead of being generated interactively.}.

To measure the key derivation delay, we run the Python app on a laptop with an Intel CPU, model i7-8565U and 16GB of RAM running Ubuntu 18.04, while the phone app runs on a Samsung Galaxy S9+ with Android 10. Both devices are connected via WiFi to either (i) a university network; or (ii) a home network. 
As a benchmark, we perform 200 encryption and decryption operations on different file sizes.
We select files of size 100KB, 1MB, 5MB, 10MB and 100MB. We intentionally select files that are relatively small (e.g., compared to video recordings) because the overhead of \name is more noticeable. 

\subsection{Results}

Table~\ref{tab:exp} shows the interactive key generation delay during encryption and decryption of a single file\footnote{The communication overhead of key derivation could be optimized using simple techniques like batching multiple key derivations together, while computational overhead could be reduced with a multi-threaded implementation.}.
Deriving a single encryption key takes 380 ms in both networks, while key derivation during encryption takes 221 ms.
As expected, standard deviation on the university network is lower -- we attribute this to higher quality network equipment.
\name introduces some overhead due to the interactive key derivation function: only about 30\% of the time is spent on computation during encryption, respectively, 50\% during decryption, while the rest of the time is spent communicating data between the devices.
However, this overhead is small compared to the fixed costs of file upload and download, as we show in the following.

Figure~\ref{fig:dropbox_eval} shows file upload and download times for \name using Dropbox as a storage backend. The light blue portion of each bar is the data upload/download time, green part indicates the encryption/decryption time, and the purple segment is the aforementioned key generation time. The overhead of our solution w.r.t. an unencrypted cloud storage solution such as Dropbox is the sum of encryption/decryption, and key generation time.
As expected, key derivation time is independent of the file size, while data upload/download times increase for larger files. 
In our worst-case, represented by encryption, for small 100KB files like text documents the overhead of our solution is approximately 20\%. For medium-sized 5MB files like high-resolution pictures the overhead goes down to 15\%, and becomes a negligible 5\% for larger files of 100MB.

\begin{table}[t]
	\centering
	\footnotesize
	\caption{Key derivation performance evaluation.}
	\label{tab:exp}
	\begin{tabular}{llcccc}
		&   & Time(ms) & std.dev & Computation(ms) & std.dev \\ \toprule
		\multirow{2}{*}{\rotatebox[]{90}{Home}} & Enc. & 380.70   & 94.77   & 118.25          & 21.6    \\  
		& Dec. & 221.48   & 76.95   & 109.84          & 13.16   \\ \midrule
		\multirow{2}{*}{\rotatebox[]{90}{Uni.}} & Enc. & 380.50   & 62.65   & 123.63          & 20.52   \\  
		& Dec. & 210.16   & 58.26   & 110.52          & 12.67   \\ \bottomrule
	\end{tabular}
\end{table}

\begin{figure}[t]
	\centering
	\includegraphics[width=0.49\linewidth]{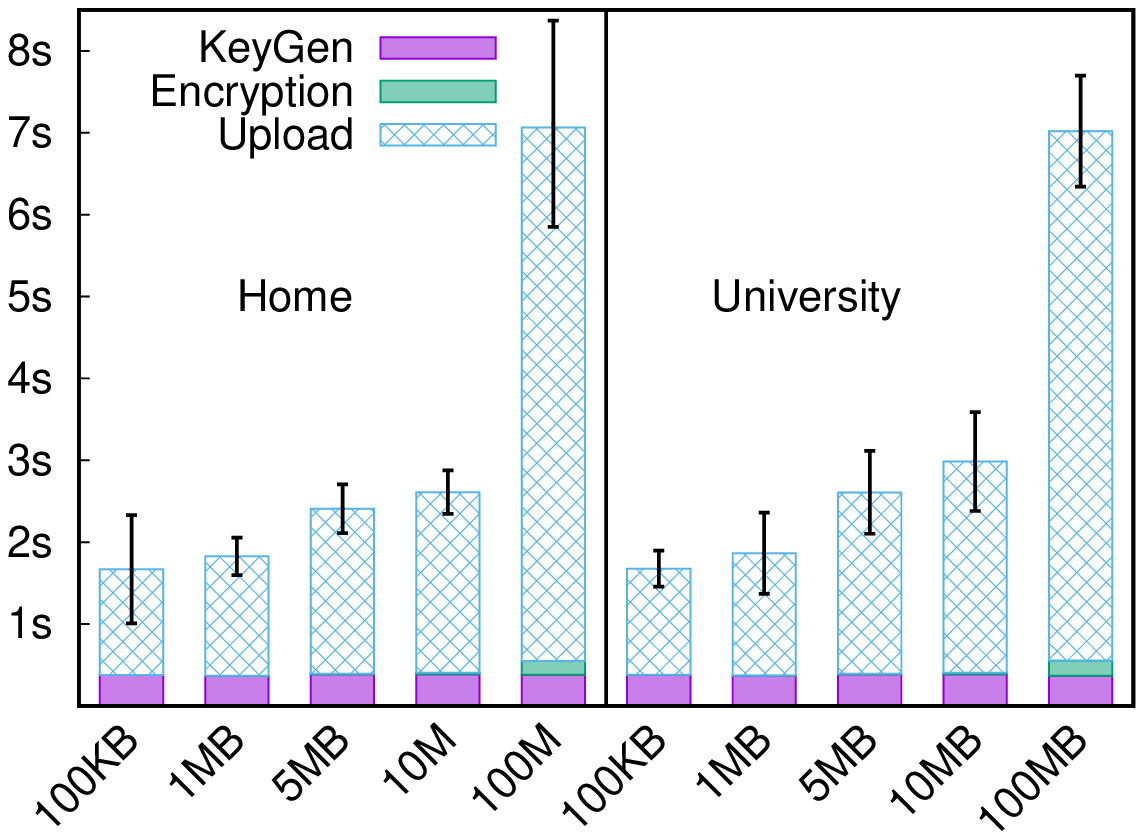}
	\includegraphics[width=0.49\linewidth]{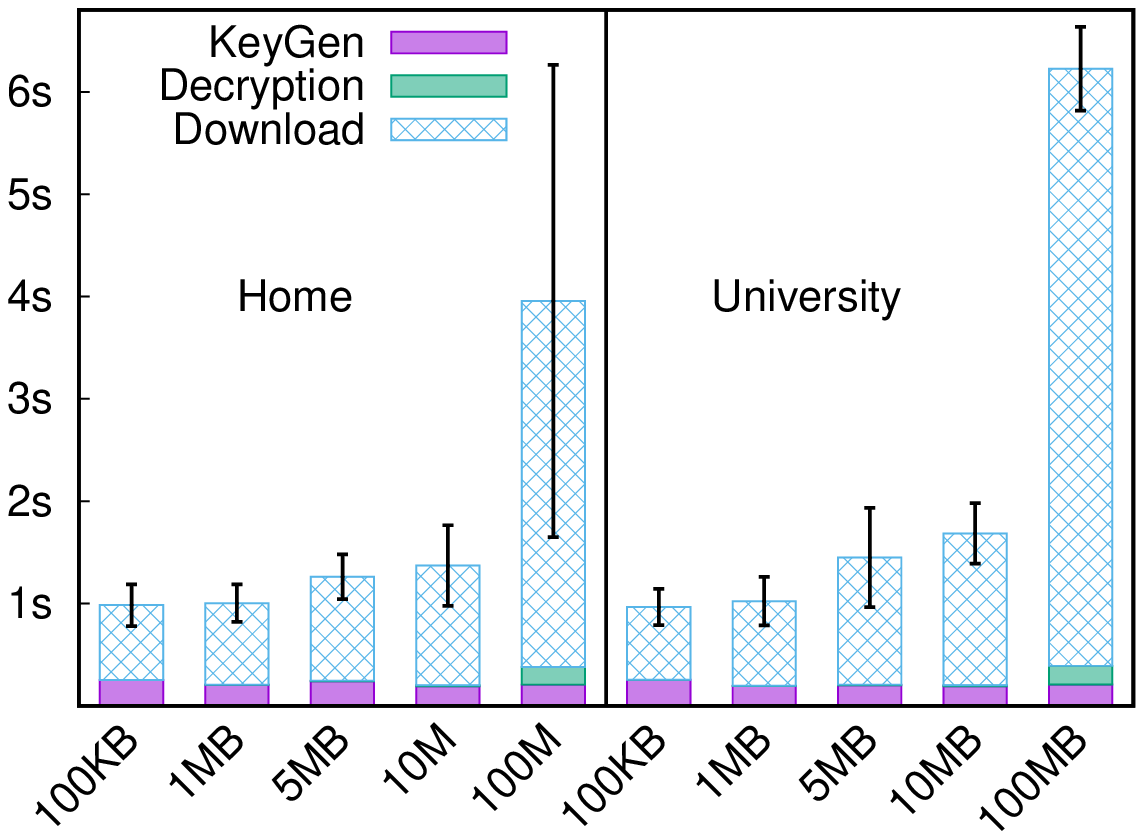}
	\caption{Time spent for key generation, encryption/decryption, and data upload/download, when storing (left) and retrieving (right) files from Dropbox using \name.	}
	\label{fig:dropbox_eval}
\end{figure}


\section{Generalization}
\label{sec:cw}

The main benefits of \name are strong confidentiality and high availability. Obviously, both of them are desirable properties in other applications as well, and in this section we discuss how to adapt our technique to other use cases. 
We start by observing that, on a high level, \name contains two distinct components. The first component manages key shares in such a way that no single device knows the full key, but the user can still recover the key should they lose access to a device (with the help of the cloud provider). The second component derives file encryption keys for cloud storage using the key shares managed by the first component.\footnote{This observation was in fact already implicitly made when analyzing the security of the recovery phase of our protocol in Section~\ref{sec:security-recovery}, where we did not rely on how encryption was performed.}

One benefit of our design is that by replacing the second component we can \emph{re-purpose} our solution to other applications while retaining the strong confidentiality and availability guarantees provided by the first component. For example, we can replace the encryption key derivation with a threshold signature scheme, allowing to use our approach to securely manage signing keys. We discuss one such example next. 

\subsection{Example Use Case: Crypto Wallets}
\label{sec:crypto-wallets}

Safe management of cryptocurrency credentials, such as Bitcoin private keys, has become increasingly important. The system that is responsible for the management of user's cryptocurrency credentials is typically called a wallet. 

\paragraph{Wallet types.} Popular wallets can be divided into the following three categories.
The first type is \emph{online wallet} (or web wallet) where the user's credentials are stored by an online service provider that typically operates a currency exchange. The main problem of online wallets is obvious: the service provider needs to be trusted and a breach can mean complete loss of funds for all users. The use of online wallets is generally not recommended for significant funds. Indeed, in 2019 alone more than ten major exchanges where breached resulting in stolen funds worth of hundreds of millions of US dollars and affecting thousands of users~\cite{exchange-hacks}.

The second type is \emph{software wallet} (sometimes called hot wallet) where the user's cryptocurrency credentials are stored locally on a standard PC. The main problem of software wallets is malware. If the user's platform becomes infected, the malware can modify transactions or steal credentials which can mean full loss of all funds. Such malware has become popular, e.g., in 2018 alone one such malware instance was spotted in more than 250,000 user platforms~\cite{wallet-malware}.

The third is \emph{hardware wallet} (also known as cold wallet) where the credentials are stored on a separate small device that can be connected to the user's PC. Hardware wallets have obvious security benefits over the previously listed approaches, as they do not require a trusted online service provider nor are they susceptible to malware on the user's PC. Consequently, hardware wallets are quickly becoming the \emph{de facto} way of managing credentials that control to any significant amount of cryptocurrency funds. For these reasons, we focus our discussion on hardware wallets.

\paragraph{Problems of hardware wallets.}
Hardware wallets provide strong confidentiality for private keys stored inside them, but their practical use has also significant issues. One problem is poor availability: should the hardware device get lost or damaged, all funds are irrecoverably lost without a proper backup. The available backup mechanisms are cumbersome and error prone. For example, current market leader products like Ledger Nano and Trezor One show a \emph{recovery sentence} that consists of 24 words to the user at the time of initialization~\cite{trezor-wiki,ledger-recovery}. The user is expected to copy this sentence on a piece of paper and keep it safe. As such process is prone to errors, some brands like Ledger recommend confirmation procedures like typing all the 24 words back to the device to eliminate possible errors (a tedious task on small device with only two buttons). If such backup processed is skipped, a human error takes place during it or the piece of paper gets lost, all funds are gone. The importance of wallets with good availability is highlighted by recent studies that estimate that up to one fifth of all Bitcoins are permanently inaccessible~\cite{wsj-missing-bitcoin}.

Another problem with hardware wallets is poor \emph{transaction confirmation security}. In a typical hardware wallet, the transaction details are received from the user's PC that provides a rich UI. Since the PC platform may be infected with malware, the user is expected to confirm the transaction details, such as recipient's address and transferred amount, using the display of the hardware wallet device. Popular devices, like Ledger Nano S, have a tiny display (30x128 pixels) that make verification of a typical cryptocurrency address (32 hexadecimals in the case of Bitcoin) tedious. Previous studies in the context of online banking have shown that transaction confirmation process can susceptible to attacks where parts of the payment details have been modified by malware~\cite{haupert2019}. The fact that cryptocurrency payments need to be confirmed using very restricted UIs is likely to increase the success rate of such attacks.

To summarize, although hardware wallets are arguably the best available way to store cryptocurrency credentials at the moment, the currently available solution suffer from poor availability, transaction insecurity and bad usability.

\subsection{Two-Factor Wallet (\walletname)}
\label{sec:tsig}

Next, we present a new solution based on the our approach to address the problems listed above. We call this solution Two-Factor Wallet (\walletname). As already hinted, the main idea is to combine the key share management component with an off-the-shelf threshold signature scheme. 

Recall that a $(n,t)$ threshold signature protocol for a signature scheme $(\Sig,\Ver)$ is an interactive protocol $\Pi_{\Sig}$ where $n$ participants each hold a part of the signing key $\sk$, and where at least $t$ must participate in order to sign a message.%
\footnote{A threshold signature scheme also comprises a threshold key generation procedure that we omit here since the master key is not generated by the system, but provided as input during enrollment.}
Essentially, a threshold signature scheme is a protocol which allows two or more devices to jointly sign a document, without any single device being able to do so on their own.%

\paragraph{Design outline.}
To keep the present presentation in line with the previous sections, consider a functionality for threshold signatures presented in Functionality~\ref{func:tsig}.
This functionality can be realized by any of the many recent protocols for threshold signatures with active security, such as \cite{DBLP:conf/sp/DoernerKLS18,DBLP:conf/esorics/DalskovOKSS20,DBLP:conf/ccs/GennaroG18,DBLP:conf/crypto/Lindell17} for ECDSA --- the most used signature scheme for cryptocurrencies. 

\begin{functionality}{func:tsig}{$\tSIG(\cdot)$---Threshold Signature}
  Let $(\Sig,\Ver)$ be a signature scheme (for simplicity, we ignore key generation).
  \begin{itemize}
  \item On $\cmda{init}{id, m,\keyC}$ from the primary device, store $(id, m, \keyC)$ and wait.
  \item On $\cmda{init}{id',m',\keyD}$ from the secondary device and if $(id',m,\keyC)$ was previously stored:
    \begin{itemize}
    \item If $m\neq m'$, set $\sigma=\bot$.
    \item If $m=m'$ compute $\sigma=\Sig(\keyC+\keyD,m)$
    \end{itemize}
  \item Output $\sigma$ to the primary device.
  \end{itemize}
\end{functionality}

We now use Functionality~\ref{func:tsig} to design \walletname. 
We now illustrate how \walletname works when the user wants to make a transaction.
We assume that system enrollment already happened (cf.~Figure \ref{fig:twofe-init}) albeit with the small twist that the public key associated with the user's private key is also distributed towards all parties.%
\footnote{For ECDSA, the primary device would send a $\pk$ value just as the secondary device does, and the real public key would be the product of the two.}
We let $\pk$ denote this value, and $m$ the message that the user wants to sign (i.e., the transaction):
\begin{enumerate}
\item The user enters $m$ into the primary device which forwards it to the secondary device.
\item The secondary device displays $m$ towards the user for approval.
  Note that this is essentially the ``prompt'' feature described in Figure \ref{fig:twofe-dec}.
\item Assuming the user approves the message on the secondary device, the primary and secondary device both enter their key-share and message into Functionality \ref{func:tsig}.
\item When the primary device receives $\sigma$, it checks if $\sigma=\bot$ and displays an error if this is the case; otherwise it outputs $\sigma$.
\end{enumerate}

\paragraph{Brief analysis.}
%
To analyze 2FW, we consider the same threat model that was outlined in Section~\ref{sec:2fe_overview}. 
Regarding confidentiality, as nothing was changed about how keys are generated, all confidentiality guarantees with respect to the signing key is preserved. 
By definition Functionality \ref{func:tsig} does not leak any of the key-shares, and therefore nothing about the signing key.
Functionality \ref{func:tsig} can be instantiated by any number of actively secure threshold signature protocols.
\walletname is guaranteed to sign only the right message since the user checks the message to be signed on both their devices.
If the user proceeds with signing, and since one of the devices is assumed to be honest, the $\mathsf{tSIG}$ Functionality is guaranteed to output either an error ($\bot$) or a signature on the correct message. Notice that for this to be true, it is important that the instantiation is actively secure.

Regarding availability, \walletname provides exactly the same guarantees as \name. If one of the user device is lost of stolen, the user can still recover his cryptocurrency credential. This is a significant improvement over current hardware wallets.

Finally, we consider usability and transaction confirmation security. Compared to hardware wallets, the main usability benefit of \walletname is that the user does not need to perform complicated and error prone backup procedures. Also safe transaction confirmation is made easier to the user, since the payments can be confirmed from a reasonably-sized smartphone screen (on the secondary device) opposed to a tiny hardware wallet display. Previous studies in the context of mobile banking have indicated that a larger user interface enable more structured transaction presentation and reduce confirmation errors~\cite{haupert2019}.


\section{Related work}
\label{sec:related-work}

\paragraph{Encrypted cloud storage.}
Prior works~\cite{adya2003farsite, goh2003sirius, wang2009ccsw} propose systems that transform a public storage service into an encrypted one by storing only encrypted files remotely.
Kamara and Lauter~\cite{DBLP:conf/fc/KamaraL10} present a solution for encrypted cloud storage in an enterprise setting and introduce a new component that manages the key material.
The main drawback of these solutions is that they consider only the storage provider as an attacker, while user devices are fully trusted.
Tang et al.~\cite{tang2012fade} propose a system that acts as an overlay on top of traditional storage providers and makes use of a quorum of key servers to assure that files are accessed and deleted (by erasing the key) according to the user's policies.
In contrast, we assume no trust on any storage or key server. 

\paragraph{Multi-devices solutions.}
The works of Liu et.~al \cite{liu2015two}, and later Zuo et.~al \cite{zuo2017fine}, focus on \emph{file sharing} with a notion of multi-factor security.
The main goal of these schemes is to allow user $A$ to encrypt a file towards another user $B$, such that $B$ can only recover the file with a key \emph{and} a hardware token.
Both works encrypt the file twice, where the first encryption is performed at $A$ and the second at the storage server.
The main difference to our work is that such schemes require a trusted key distribution mechanism, and thus a trusted third party. Our solution assumes no trusted third parties. 
In addition, since these works focus on file sharing, they do not solve the cloud storage problem. 

Similar to these works, \cite{7414154} presents a solution for two factor encryption also utilizing hardware tokens, where first a random string is sampled and a key is derived on a commodity hardware token.
The client then derives the final key from the such key and a password.
This approach does not allow one to control if a key is derived or not, in case the client is malicious.
Nor does it protect against a user who picks a bad password (such as the users we consider) and an adversary who compromises the hardware token.

Another prior work \cite{8284478} explore using Shamir secret-sharing to store an encryption key on multiple devices.
However, this work does not address key derivation; nor does it consider how the user would actually go about storing all these shares.

Omnishare \cite{8325583} is another recent system which allows multiple devices to access the same files from encrypted cloud storage. Omnishare is based on key duplication which is vulnerable to device compromise, as explained in Section~\ref{sec:approach}.

\paragraph{Oblivious key management.}
Our threshold PRF construction leverages the oblivious PRF protocol from \cite{DBLP:conf/asiacrypt/JareckiKK14} (see Appendix \ref{sec:2fe}).
The same oblivious PRF protocol is adopted in several other use cases as well, such as in password-protected secret-sharing \cite{jarecki2016highly},\cite{DBLP:conf:acns:JareckiKKX17}, password management \cite{shirvanian2017sphinx}, or key management services for large-scale storage solutions \cite{10.1145/3319535.3363196}.


\section{Conclusions}
\label{sec:conclusions}

The level of security provided by the existing encrypted cloud storage solutions has been unclear, since most commercial service providers do not publish detailed security analysis with well-defined threat models and security properties. To address such gap, we have defined an extensive but realistic threat model for encrypted cloud storage and analyzed 10 popular commercial services. We conclude that none of them provides sufficient confidentiality and availability. 

Motivated by such lack of secure solutions in the market (and research literature), we have designed a novel solution called Two-Factor Encryption (2FE), where two user’s devices, such as laptop and smartphone, interact in file encryption and decryption, in the spirit of two-factor authentication. 2FE tolerates a broad range of threats that include untrusted cloud storage, device theft and compromise, and various benign human errors. We believe that 2FE is by far the most secure and practical encrypted cloud storage solution known today.

Finally, our flexible construction easily adapt to other use cases---we show it by presenting a novel design for a secure cryptocurrency wallet.


{\small
\bibliographystyle{plain}
\bibliography{refs}}

\begin{appendix}
  \section{Instantiations}
\label{sec:2fe}

In this appendix we provide  instantiations and security arguments for the primitives and functionalities that were defined in Section~\ref{sec:functionalities} and used in our solution design in Section~\ref{sec:full-protocol}. 

\subsection{Secret Sharing with Updates}
\label{sec:enrollment}

We start by describing the specifics of the secret sharing that is compatible with our instantiation of the $\tPRF$ functionality (described below) and that supports key share updates used in our solution.

The primary and secondary each generate at random a key share $\keyC$ and $\keyD$, both elements of the finite field $\Z_{p}=[0,\dots,p-1]$.
Note that this defines the master encryption key $\PRFkey=\keyD+\keyC$ where addition is modulo $p$.
The secondary device sends a public key $\pk=\textsf{Gen}(\keyD)=\keyD\cdot G$ where $G$ is a generator of $\mathcal{G}$, a subgroup of order $p$ of an elliptic curve group $E$.
We assume that the Decisional Diffie-Hellman (DDH) assumption holds in $\mathcal{G}$ (see Definition \ref{def:ddh}).

Recall from Section \ref{sec:challenges} that we need to refresh all key material after each recover process.
Otherwise, compromises during different times of operation would eventually leak all of $\PRFkey$.
Refreshing keys is performed by first adding a secret-sharing of 0 to each key share, using the aforementioned property that we can compute additions over shares. 
Say, $\C$ picks a random value $z$ and then updates its key to be $\keyC'=\keyC+z$ and sends $z'=-z$ to $\D$ who sets $\keyD'=\keyD+z'$.
Notice that $(\keyC',\keyD')$ still defines a sharing of $\PRFkey$.

\paragraph{Security arguments.}
Clearly, neither the secondary device knows nothing about the primary device's key share as this was generated entirely locally.
The primary device knows something about the secondary's key-share, namely $\pk=\keyD\cdot G$.
However this reveals nothing as the discrete log problem is assumed to be hard for the group $\mathcal{G}$ (a consequence of assuming DDH holds in $\mathcal{G}$), and because $\keyD$ was picked at random by the secondary device.

To show that key share updates are secure we need to argue the following: $\PRFkey$ remains hidden from an adversary $\adv$, so long as (1) $\adv$ only corrupts one device; and (2) when $\adv$ corrupts a different device, then we run the key recovery protocol in between corruptions \emph{without} involvement from $\adv$.
(1) clearly holds from the same argument we used in the previous section; that is, $(\keyC,\keyD)$ is a 2-out-of-2 secret sharing of $\PRFkey$, and so $\adv$ can learn one of $\keyC$ or $\keyD$ without learning anything about $\PRFkey$.
Now consider (2) and let $\keyD$ be a share learned by the adversary \emph{before} the key refresh protocol, and $\keyC=\keyC+z$ be a share learned \emph{}after the key refresh protocol.
Crucially, because the adversary does not participate in the recovery protocol, $z$ is \emph{unknown} to $\adv$.
In particular, $\keyD+\keyC'=\PRFkey+z$ reveals no information about $\PRFkey$.

\subsection{Shared Randomness Protocol}

Our solution needs a way to generate a uniformly random bit-string that is the same for both parties (Functionality~\ref{func:cointoss}). Next, we describe how such protocol can be constructed using a commitment scheme which we build from a hash function. This is shown below in Protocol \ref{prot:ct} that relies on the fact that $s_{0}$ is a high entropy string.
If this was not the case, the additional randomness needs to added to the commitment in order to ensure hiding.

\begin{protocol}{prot:ct}{$\SR$ --- Shared Randomness}
\begin{enumerate}
\item $\C$ samples $s_{0}\pick\bits^{\kappa}$ and sends the commitment $c=H(s_{0})$ to $\D$, where $H$ is a hash function modeled as a Random Oracle, and $\kappa$ a security parameter (e.g., $256$).
\item $\D$, on receiving $c$, picks $s_{1}\pick\bits^{\kappa}$ and sends $s_{1}$ to $\C$.
\item $\C$ when it receives $s_{1}$, sends $s_{0}$ to $\D$ and outputs $s=s_{0}\oplus s_{1}$.
\item $\D$ when it receives $s_{0}$, checks that $c=H(s_{0})$ and if so outputs $s=s_{1}\oplus s_{0}$.
\end{enumerate}
\end{protocol}

\paragraph{Security argument.}
Let see why the $\SR$ protocol is secure, that is, why the output is cannot be biased unreasonably.
Consider first the view of the secondary $\D$: the commitment $c$ received from the primary $\C$ reveals nothing about $s_{0}$, and so $\D$ cannot pick $s_{1}$ which would bias the output towards a particular value.
On the other hand, $\C$ can introduce a slight bias in the output.
The issue is that protocols involving only two parties are inherently unfair. That is, the primary $\C$ gets to see the result $s$ before the secondary $\D$ and can thus decide to abort if it does not like the result.
However, this bias is bounded \cite{cleve1986limits} and does not pose a practical problem for large outputs.

\subsection{Threshold PRF Protocol}
\label{sec:key-derivation}

Finally, our solution needs a threshold PRF protocol defined as Functionality \ref{func:tprf}. This protocol lets the primary input $\keyC$, the secondary inputs $\keyD$ and both inputs some additional information.
The output of $\tPRF$ is either a symmetric key $k$ given only to the primary, or an error.
In case no error is reported $k$ is guaranteed to have been computed ``correctly''. Basically, the primary $\C$ gets a guarantee that the secondary $\D$ did not misbehave.

We obtain such $\tPRF$ protocol by performing a slight alternation to the 2HashDH oblivious PRF in \cite{DBLP:conf/asiacrypt/JareckiKK14}.
We relax the ``obliviousness'' requirement of the construction, and reveal the input (the string $x$ in Protocol \ref{prot:tprf} below) to both devices.
This relaxation is needed to support prompting on the secondary device during decryption.
Our $\tPRF$ protocol proceeds as shown in Protocol \ref{prot:tprf} below.

\begin{protocol}{prot:tprf}{$\tPRF$ --- Threshold PRF}
\begin{enumerate}
\item Let $x\in\bits^{*}$ be the input, given to $\C$ and $\D$.
\item $\D$ computes $B= \keyD\cdot (H'(x))$ and a proof of correctness $\pi=\nizkDDH.\mathsf{Gen}(G,\keyD,H'(x),B)$. Send $(B,\pi)$ to $\C$.
\item When $\C$ receives the tuple $(B,\pi)$ first check if $\nizkDDH.\mathsf{Ver}(G,H'(x),B,\pi,\pk)=1$ and output $\mathtt{error}$ if this is not the case.
  Otherwise, $\C$ computes and outputs key $k$ as $k= H(x,B+\keyC\cdot (H'(x)))$.
\end{enumerate}
\end{protocol}

Notice that if the output is not $\mathtt{error}$, then with high probability $k=H(x,\PRFkey\cdot(H'(x)))$.

\paragraph{NIZK proofs.}
The above protocol requires \emph{non-interactive zero-knowledge} (NIZK) proofs for equality of discrete logs. Let $E(\Z_{q})$ be an elliptic curve over $\Z_{q}$.
A  NIZK proof for the statement $\log_{A}(B)=\log_{G}(P)$ is a tuple $\nizkDDH=(\mathsf{Gen},\mathsf{Ver})$ where \cite{10.1007:3-540-48071-4_7}:
\begin{itemize}
\item $\nizkDDH.\mathsf{Gen}(G,x,A,B)$ with $A=xB$, first picks a random $t$ and then computes $w=H(G,xG,A,B,tG,t A)$ where $H$ is a hash function.
  Output proof $\pi=(w,s)$ where $s=t+w\cdot x$.
\item $\nizkDDH.\mathsf{Ver}(G,A,B,\pi,P)$ with $P=x G$.
  Parse $\pi$ as $(w,s)$ and compute
  $$w'=H(G,P,A,B,s G-w P, s A-w B).$$
  Output $1$ if $w'=w$ and $0$ otherwise.
\end{itemize}

\paragraph{Security argument.}
For the subsequent security arguments for it to be useful in general for key generation, $\tPRF$ must satisfy the following criteria:
\begin{enumerate}
\item The secondary device must behave honestly.
\item No information about $\keyD$ must be leaked to the primary, and no information about $\keyC$ must be leaked to the secondary.
\item The output must be indistinguishable from a random string.
\end{enumerate}

\paragraph{\emph{Point 1.}}
This follows immediately from the security of the NIZK.
Suppose the secondary is able to cheat, i.e., it sends $B'\neq\keyD\cdot(H'(x))$ and $\pi=(w,s)$ that is accepted by the primary.
Observe that the check that the primary performs can be written as
\begin{align*}
  w &= H(G,P,A,B',sG-wP,sA-wB') \\
    &=H(G,P,A,B',rG,r'A).
\end{align*}
where $r=(s-w\keyD)$ and $r'=(s-wx')$ for some $x'\neq\keyD$.
In particular, $r$ and $r'$ could be computed by the secondary.
But, since $x'\neq\keyD$ we have that $w=(r-r')/(x'-\keyD)$, we can conclude that such $r$, $r'$ cannot be found unless with probability $1/2^{\kappa}$ where $\kappa$ is the length of the output of $H$.

\paragraph{\emph{Point 2.}}
That no information about $\keyC$ is leaked to the secondary device is obvious, as no message involving $\keyC$ is ever sent to the secondary device.
That no information about $\keyD$ is leaked to the primary device requires us to prove the following:
Suppose the primary receives values $B_{1},B_{2},\dots,B_{n}$ where $B_{i}=\keyD G_{i}$ for some known $G_{i}$.
We must now argue that these values reveal no information about $\keyD$; in particular, that this sequence of $B$'s are indistinguishable from a series of random group elements.
To that end, we define the following experiment
\begin{enumerate}
\item Pick $x$ and send $P=xG$ to the adversary. Let $c\in\{0,1\}$ be a randomly chosen bit.
\item In a loop, the adversary sends $t$ and receives back $B=x(H'(t))$.
\item The adversary finally sends $t'$. If $t'$ was sent in the previous step, output $c$ and stop.
  Otherwise, if $c=0$, compute $B'=x(H'(t'))$ and if $c'$ pick $B'$ as a uniformly random group element.
  Return $B'$ to the adversary.
\item The adversary outputs $c'\in\{0,1\}$ and wins if $c'=c$.
\end{enumerate}

We will make use of the following computational assumption.

\begin{definition}{Decisional Diffie Hellman (DDH).}
  \label{def:ddh}
  Let $Z_{p}$ be a group with generator $G$.
  For any PPT adversary $\adv$, we have that
  \[
    |\Pr[\adv(uG,vG,wG)=1] - \Pr[\adv(uG,vG,uvG)=1]|
  \]
  is negligible, where $u,v,w$ are picked at random from $[0,\dots,p-1]$.
\end{definition}
In a nutshell, the above states that there is no efficient program for distinguishing between $wG$ and $uvG$, given $uG$ and $vG$.

We now prove the following.

\begin{lemma}
  Any adversary wins the above game with at most negligible probability assuming the Decisional Diffie-Hellman (DDH) problem is hard.
\end{lemma}

\begin{proof}
  Let $(uG,vG,wG)$ be a DDH challenge and $\mathcal{A}$ and adversary who wins the above game with probability $\epsilon$.
  Set $P=uG$ in the experiment (this effectively sets the ``key'' to be $u$).
  Let $q$ be an upper bound on the number of queries that $\mathcal{A}$ makes to $H'$ before issuing it's challenge $t'$, let $j\in[q+1]$.
  Set $k=0$, $T=\emptyset$ and $\alpha'=\bot$.
  Every time $\mathcal{A}$ queries $H'$ on $\alpha$, if $(\alpha, \beta)\in T$, return $\beta G$.
  Otherwise, set $k=k+1$. If $k=j$, set $\alpha'=\alpha$ and $H'(\alpha)=vG$.
  If $k\neq j$, pick $\beta$ at random, set $H'(\alpha)=\beta G$ and $T=T\cup\{(\alpha,\beta)\}$.
  For each value $t_{i}$ sent by $\mathcal{A}$ in step 2 do as follows:
  If $(t_{i},\beta)\not\in T$ for some $\beta$, do as above, and return $B_{i}=\beta (uG)=u(\beta G)=u(H'(t_{i}))$.
  Let $t'$ be the challenge $\mathcal{A}$ sends in step 4.
  If $t'$ was among the $t_{i}$'s, or if $t'\neq\alpha'$ and $\alpha'\neq\bot$ output a random bit and stop.
  If $\alpha=\bot$, set $H'(t')=vG$. Return $B'=wG$.
  When $\mathcal{A}$ outputs $c'$, output $c'$ as well.

  Notice that, with probability at least $1/(q+1)$ the reduction correctly guesses which oracle query $\mathcal{A}$ will use as its challenge.
  Let $c'$ be the output of $\mathcal{A}$ in this case.
  If $c'=0$, then $\mathcal{A}$ believes that $B'=u(H'(t'))=u(vG)=wG$, while if $c'=1$ $\mathcal{A}$ will believe that $B'$ was random.
  This correspond exactly to whether or not $(uG,vG,wG)$ was a DDH tuple or not.
  Thus the reduction wins the DDH game with probability $\epsilon/(q+1)$, which is negligible.
\end{proof}

\paragraph{\emph{Point 3.}}
Now this last point---that the output of $\tPRF$ is indistinguishable from random follows from Lemma 2.
In particular, the primary device cannot compute $H(x,\PRFkey(H'(x)))$ himself and so in particular cannot see whether the output came from such a computation or whether it was picked at random.


  \section{Further Analysis Remarks}
\label{sec:analysis-remarks}

We provide further observations and remarks about the commercial encrypted cloud storage services that we analyzed in Section~\ref{sec:motivation}, and motivate their scores.

\paragraph{Tarsnap.}
Tarsnap does not provide conventional username-password authentication, but stores a local file on the user's device that contains per-file encryption keys for each encrypted file and access tokens that used to authenticate connections to the cloud storage.

\paragraph{Mega.}
Mega supports long-lived sessions; it also supports remote session termination, which can be used to disable long-lived sessions on a lost device (see Section 3.6 in \cite{mega-whitepaper}).
A 256-bit master key is derived from the user's password; the first half is used to encrypt files, while the second half is used for authentication.
The master key itself is static (i.e., generated once and never changed) and is further stored encrypted under the user's password.

For availability, Mega offers a trash-bin, although users can empty their bins, thus losing such files forever.
Users can recover from the external adversary changing their passwords by providing a cold copy of their master key, that Mega provides.

\paragraph{pCloud.}
pCloud encrypts files with a 256-bit key, password-protected with a ``crypto password''.
This password has to be entered by the user at every session---thus, long-lived sessions are not possible, and confidentiality and availability to the device stealing external adversary are provided.
Further, pCloud supports password hints, stored in plaintext in their servers.

Availability to more powerful external adversaries is not provided, because despite providing a trash-bin functionality, the encryption password can be changed from any device.

\paragraph{Sync.}
Sync, similarly to other password-based solutions, cannot send the authentication password to the server, as it would leak the encryption keys. 
However, Sync sends a weak bcrypt hash during authentication: one with only $256=2^{8}$ iterations---thus, bruteforce is possible for the cloud adversary.
It supports long-lived sessions, both in the browser and through their software.
It also supports an optional ``email-based password recovery'' feature that, when users opt-in, leaks the keys to the server.
Finally, Sync allows users to store unencrypted ``password hints'' on their servers.

Availability of Sync is good: it supports file versioning up to 1 year, thus protecting against malicious file deletions.
Further, if the user is logged in more than one device, if any external adversary would try to change the user's password these would remain logged in.

\paragraph{Woelkli.}
Woelkli protects a key file with the user password.
However, data is not end-to-end encrypted, as encryption happens on the server side---further confirmed by the fact that the web interface sends the user's password in clear during login.
This means that for the duration of a user session, confidentiality to the cloud provider cannot be guaranteed.
Further, it supports long-lived sessions, thus not providing confidentiality to the external adversary.

For availability, Woelkli provides a trash-bin, although with an emptying functionality that permanently removes the files from the cloud.
Further, if the password is changed on one device, other user devices cannot access anymore.

Woelkli, similarly to Sync, advertises and invites users to activate \textit{email-based password recovery} if they are worried about forgetting their passwords---however, to allow this feature, keys would be leaked to the service, losing all confidentiality guarantees to the cloud storage adversary.

\paragraph{SpiderOak.}
SpiderOak stores the user password, needed to decrypt the file encryption keys, in cleartext or as a MD5 hash on the client~\cite{DBLP:conf:ccs:DalskovO18}.
Further, it supports long-lived sessions, a threat to confidentiality against the external adversary.
Using the SpiderOak web interface leaks the encryption keys to the server, and SpiderOak allows storing a plaintext ``password hint'' on their servers.

For availability, SpiderOak supports file versioning, as a safeguard against malicious file deletion.
However, long-lived sessions and the ability to close other open sessions and change the password mean that the external adversary can lock out the legitimate user.

\paragraph{Tresorit.}
Tresorit allows configurable session length---long-lived sessions hamper confidentiality to the external adversary.
Further, keys are kept in plaintext on the user's device.

For availability, Tresorit regenerates encryption keys when the password is changed---thus, the external adversary could lock the legitimate user out.

\paragraph{Boxcryptor.}
Boxcryptor allows long-lived user sessions, that harm confidentiality against the external adversary.
User's encryption keys are password-protected with the user's password.

Availability of Boxcryptor is good under some conditions: it is dependant on the cloud storage backend provider's offer of file versioning.
Further, with a cold copy of the encryption key, users can restore access to their account even if the external adversary changed the password from the stolen or controlled device.

\paragraph{Zoolz.}
In Zoolz encryption can be configured during setup, where an encryption key can be generated from a user-selected password.
Once a key has been set, it cannot be changed later.
It supports long-lived sessions.
Deduplication is performed on encrypted files \emph{if} they are protected with the same password.

For availability, an adversary that has control of a user device with an open session can delete their data permanently.
However the attacker cannot easily change the user's password for existing data.

\paragraph{IDrive.}
IDrive supports setting a ``private encryption key'', although this is not the default.
If such a key is provided, then encryption (with the private encryption key) is performed on the device before data is sent to the server.
This key can only be changed by resetting the account, which wipes all data (or makes it impossible to decrypt).
The private encryption key is derived from a string picked by the user, and a small bit of data is stored on the server encrypted under this key for use during authentication (or to validate an entered key).
However, web access leaks keys and users' files to an ``intermediate'' IDrive server.
IDrive supports long-lived sessions, which means the keys have to be in plaintext in the local device.

For availability, IDrive supports file versioning.
However, if the external adversary changes the user's password, this automatically wipes all user's files, thus hampering availability. This is a threat because of long-lived sessions.


\end{appendix}

\end{document}